\documentclass[11pt,fleqn]{article}

\usepackage{amssymb,amsmath}
\usepackage{natbib}
\usepackage{vmargin,setspace,verbatim}
\usepackage{graphicx,graphics,color,tikz,tikz-network}
\usepackage{microtype}

\usepackage{libertine}
\usepackage[libertine]{newtxmath}

\definecolor{darkgreen}{rgb}{0,0.2,0}
\definecolor{darkred}{rgb}{0.3,0,0}

\usepackage[colorlinks=true, pdfstartview=FitV, linkcolor=darkgreen, citecolor=darkred, urlcolor=black]{hyperref}

\newcounter{llst}
\newenvironment{abet}{\begin{list}{\rm (\alph{llst})}{\usecounter{llst}
\setlength{\itemindent}{0em} \setlength{\leftmargin}{3em}
\setlength{\labelwidth}{2em} \setlength{\labelsep}{1em}}}{\end{list}}
\newenvironment{numm}{\begin{list}{\rm (\roman{llst})}{\usecounter{llst}
\setlength{\itemindent}{0em} \setlength{\leftmargin}{3.5em}
\setlength{\labelwidth}{2.5em} \setlength{\labelsep}{1em}}}{\end{list}}

\newtheorem{theorem}{Theorem}[section]

\newtheorem{corollary}[theorem]{Corollary}

\newtheorem{definition}[theorem]{Definition}
\newtheorem{expl}[theorem]{Example}

\newtheorem{lemma}[theorem]{Lemma}

\newtheorem{proposition}[theorem]{Proposition}

\newtheorem{dscrpt}[theorem]{Description}

\newenvironment{proof}[1][Proof]{\noindent \textbf{#1.} }{\hfill
\rule{0.5em}{0.5em}}

\newenvironment{example}{\begin{expl} \rm}{\hfill $\blacklozenge$
\end{expl}}

\setpapersize{A4}
\setmarginsrb{80pt}{40pt}{80pt}{60pt}{12pt}{10pt}{12pt}{30pt}

\begin{document}

\title{\textbf{Building social networks under consent: \\ A survey}\thanks{I am very grateful for the extensive and helpful comments and suggestions of an anonymous referee. I also thank my co-authors Sudipta Sarangi, Subhadip Chakrabarti and Owen Sims for their contributions and support to explore the paradoxes of social network formation under the hypothesis of mutual consent. The work summarised here could not have been developed without their contributions, insights and help. }}

\author{Robert P.~Gilles\thanks{Economics Group, Management School, The Queen's University of Belfast, Riddel Hall, 185 Stranmillis Road, Belfast, BT9~5EE, UK. \textsf{Email: r.gilles@qub.ac.uk}} }

\date{October 2019 \\ Revised: April 2020}

\maketitle

%\begin{abstract}
%\singlespace\noindent
%In this survey I provide an overview of recent work on the game theoretic modelling of consent in link formation.

%\end{abstract}

\section{Mutual consent in network formation}

During the past two decades there has emerged an extensive literature on game theoretic models of network formation. Seminally, the fundamentals of such a game theoretic perspective were set out by \citet{AumannMyerson1988} in which players are guided by the Myerson value of corresponding communication situations. This contribution explored network formation under mutual consent through a non-cooperative signalling game: A link between two players is formed if and only if both players signal to each other their willingness to form this relationship.  The main insight  of the \emph{Myerson model} \citep{Myerson1991} is that the network without any links is always supported through a Nash equilibrium of this signalling game. This theoretical result leads to the conclusion that network formation under mutual consent has to be considered as difficult, even impossible. This would contradict the well-established understanding of human nature as that of a social networker \citep{Seabright2004,Harari2014}.\footnote{The main conclusion is strengthened in the case of costly link formation, in which the empty network is a \emph{strong} Nash equilibrium, indicating that starting from an empty network it seems unlikely that rational agents would be able to establish non-trivial networks. I also refer to \citet{JoshiSarangi2020} for a dynamic model of such non-trivial network formation. } Nevertheless, paradoxically, the Myerson model is and remains the most natural, straightforward and convincing non-cooperative model of network formation under mutual consent.

The relative failure of this natural non-cooperative approach induced \citet{JacksonWolinsky1996} to introduce an alternative approach, which is founded on a bilateral cooperative consideration.\footnote{An alternative mathematical model emerged with \citet{BalaGoyal2000a} based on one-sided link formation: One assumes \emph{ex-ante}, or implicit, consent among players in the network formation game. The resulting equilibrium networks are denoted as \emph{Nash networks} in the subsequently developed literature. This approach is unsatisfactory due to its unnatural social foundations with rather limited applicability to explain social and economic phenomena.} In their approach, Jackson and Wolinsky allow pairs of players to cooperatively deviate from an existing network to modify it. The equilibrium networks under such pairwise modification are denoted as \emph{pairwise stable} networks. Pairwise stability provided a fertile foundation for further exploration of network formation under cooperative consent. This resulted in the development and study of variations of pairwise stability.

Although the Jackson-Wolinsky approach founded on pairwise stability has been very successful in explaining the emergence of non-trivial networks, there remained a gap in our understanding concerning a purely non-cooperative approach to the modelling of mutual consent in network formation. This has been more recently explored through the design of bespoke equilibrium concept applied in the Myerson model. In particular, \citet{Buskens2005} and \citet{GillesSarangi2010} introduced models of trusting behaviour in network formation through trust-based belief systems. The equilibrium concepts that are developed from these models have very strong properties, showing that trust in network formation leads to non-trivial equilibrium networks. For example, \citet{GillesSarangi2010}'s notion of monadic stability results in equilibrium networks that form a specific subclass of pairwise stable networks---denoted as the \emph{strictly} pairwise stable networks.

\paragraph{Overview of this survey.}

This survey explores the various methodologies to properly model mutual consent in network formation. I compare the different classes of equilibrium networks that emerge from these different methodologies. After discussing the principles of link formation under mutual consent and Myerson's seminal model, I turn to the exploration of Jackson-Wolinsky type stability concepts based on pairwise cooperative behaviour. I distinguish different subclasses of stable network based on hypotheses about how coalitions of certain sizes can modify the current network. This mainly pertains to pairs of players, but also extends to coalitions of players of arbitrary size---resulting in the notion of a \emph{strongly stable} network \citep{JacksonNouweland2005}.

Subsequently, I turn to the main non-cooperative theory of network formation under mutual consent, namely extensions of the Myerson model \citep{Myerson1991}.  I survey the results from the literature that categorise the various classes of equilibrium networks in the Myerson model with two- as well as one-sided link formation costs. There emerges a close link to certain classes of stable networks in the Jackson-Wolinsky framework. 

Subsequently, I discuss the idea of equilibrium refinement in the Myerson model to reflect considerations of mutual trust in link formation. Indeed, links are representations of socio-economic relationships that are founded on mutual trust between the interacting parties. This results in the unilateral  \citep{Buskens2005} and monadic stability \citep{GillesSarangi2010} concepts in the Myerson model. I explore monadic stability further, which is founded on a conception of mutual trust through a belief system in the Myerson model. The properties of these monadically stable networks as well as their existence, using \citet{MondererShapley1996}'s theory of game-theoretic potentials, are also reviewed.

I conclude this survey by looking at an alternative method to modelling mutual consent in link- and network formation. This refers to the introduction of correlated strategies in the Myerson model as a tool to represent coordinated interaction. The resulting class of ``correlated equilibrium networks'' still needs to be explored in future research. 

\section{Introducing mutual consent: Modelling principles}

Throughout this survey, I use a broad class of game theoretic techniques to model how relationships---or ``links''---between pairs of socio-economic agents come about. We refer to these socio-economic agents  as \emph{players} in the context of these models. Each player is assumed to  be a fully rational individual decision maker that acts according to a set of behavioural rules described in the developed equilibrium concept.

Besides the specific behavioural hypotheses on which these equilibrium concepts are based, it is important to realise that there are some fundamental broad axioms made. These fundamental axioms introduce a few fundamental limitations of the approach that is surveyed here: 
\begin{numm}
\item This game-theoretic approach is purely \emph{static} in nature. This implies that we start from a zero state in which no links exist and in which these socio-economic agents decide whether and which links to build. The end result is a fully formed network in which certain value-generating activities are achieved. It would be more realistic to model the formation of a network as a dynamic building process. However, in the static conception followed throughout this survey, one network does \emph{not} evolve into another.
\\
 This has major consequences for how we view network formation and which networks actually are identified in these constructions. Indeed, the identified equilibrium networks do not exhibit the features of large social networks identified in the literature quoted on social networks \citep{Newman2010book,Barabasi2016}. So, these equilibrium networks are usually neither scale-free nor small world networks nor satisfying the basic property of assortative mixing. This is a severe limitation of such a static approach.\footnote{In my discussion in this survey I omit the recent development of incentive-based stochastic models of network formation. This approach focuses not only on game theoretic incentives in network formation---as the subject matter of this survey---but combines this concept with stochastic processes that describe random meetings. This approach was seminally developed in \citet{JacksonRogers2005,JacksonRogers2007} and further addressed in, e.g., \citet{GolubLivne2010}.}
\\
On the other hand, the static approach highlights certain properties of rational decision making in the context of pairwise cooperation, required for building value-generating relationships under mutual consent. Rather contradictorily, the main theorem in Myerson's non-cooperative model shows that rational decision-making does actually not result into any sensible network formation---the empty network is always supported through a Nash equilibrium in the Myerson model. So, starting from an empty network, fully rational players have no mechanism to create a meaningful interaction structure. Only if we impose that the decision makers are \emph{boundedly rational}---and, thus, use animal spirits rather than optimisation in decision making---we arrive at the conclusion that non-trivial and sensible networks emerge under mutual consent.\footnote{This is captured in the notion of a monadically stable network that is founded on trusting behaviour by the players. Such trusting behaviour is fundamentally boundedly rational. Indeed, to trust another player is not founded on calculation, but on a leap of faith.} This important insight is the main conclusion presented in this survey.

\item The game-theoretic approach explored in this survey is founded on a \emph{negative} stability methodology. Hence, a network is called ``stable'' if there are no incentives to change the network. This is the standard methodology in game theory and neo-Walrasian economics. Rather than constructing an actual building process, this methodology only looks at which networks \emph{cannot} emerge due to the existing incentives to change the network that the players are endowed with. We thus arrive at a class of equilibrium networks that describe configurations in which such incentives for deviation are absent.
\\
The consequence of the application of this standard game-theoretic methodology is that reality is only approximated. This approach, for example, does not allow the mixing of modes of incentives, which is common in real-life interaction. This, therefore, is another reason why the theoretically derived networks do not have the desired features discussed in the literature on large social networks as surveyed by \citet{Barabasi2016}.
\end{numm}
The next section sets out the basic framework of modelling mutual consent in the formation of a relationship between two players. 

\subsection{Players, links and networks}

We use the basic concepts from the theory of social networks set out in the literature. Following the accepted symbolism, the set $N = \{ 1, \ldots ,n \}$ represents a set of \emph{players}. The fundamental issue addressed here is how these players will build pairwise or binary relationships with other players and ultimately construct a socio-economic network consisting of such binary relationships.

Each player $i \in N$ is explicitly endowed with the social ability to build such pairwise relationships or \emph{links} with other players, provided that consent is given by the other party. Again following the accepted terminology in the literature \citep{Jackson2008}, the pairwise subset $\{ i,j \} \subset N$ with $i \neq j$ denotes a pairwise relationship between players $i \in N$ and $j \in N$.  We follow convention to use shorthand notation and define a \emph{link} between players $i$ and $j$ as $ij = \{ i,j \} \in g^N$, where 
\begin{equation}
	g^N = \{ \, \{ i,j \} \mid i,j \in N \mbox{ and } i \neq j \} = \{ ij \mid i,j \in N \}
\end{equation}
denotes the set of all potential links on the player set $N$. As such the set $g^N$ acts as the universal set of all potential links on player set $N$.

A \emph{network} on $N$ is now an arbitrary subset of links, i.e., any subset $g \subset g^N$ is a network on $N$.  In particular, $g=g^0 = \varnothing$ is the \emph{empty network} on $N$ which describes a situation where no links are formed. Furthermore, $g = g^N$ is the \emph{complete network} on $N$, which is the largest network consisting of all potential links among players in $N$.  We introduce $\mathbb{G}^N = \{ g \mid g \subset g^N \}$ as the collection of all networks on $N$.

The \emph{neighbourhood} of player $i \in N$ in network $g \in \mathbb{G}^N$ is given by $N_i (g) = \{ j \in N \mid ij \in g \}$. The collection of corresponding neighbouring relationships or links is denoted by $L_i (g) = \{ ij \in g \mid j \in N_i (g) \}$. The complete collection of all potential links that involve player $i \in N$---or that can be formed by player $i$---is denoted by $L_i = L_i (g^N) = \{ ij \mid j \neq i \}$.

\paragraph{Adding and deleting links to a network.}

In formal models of network formation we consider the deletion and addition of links to given networks. For this I introduce some well-accepted notation \citep{Jackson2008}. Consider a network $g \in \mathbb{G}^N$. For every pair of players $i,j \in N$ with $ij \notin g$ we now denote by $g + ij$ the network that results from $g$ by adding the link $ij \notin g$, i.e., $g + ij = g \cup \{ ij \} \in \mathbb{G}^N$. Similarly, for some collection of links $h \subset g^N$ with $g \cap h = \varnothing$, we denote $g + h = g \cup h$ the network that results from adding link collection $h$ to the network $g$.

Next, consider two players $i,j \in N$ with $ij \in g$. We denote by $g - ij = g \setminus \{ ij \} \in \mathbb{G}^N$ the network that results from removing the link $ij$ from the network $g$. Again, for any collection of links $h \subset g$ we denote $g-h = g \setminus h$ the network that results from removing the links in $h$ from the network $g$.

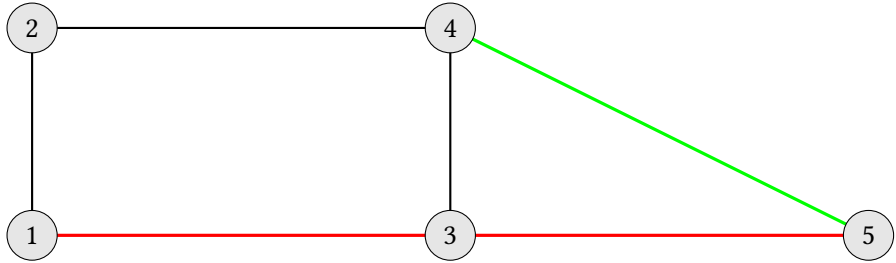
\begin{figure}[h]
\begin{center}
\begin{tikzpicture}[scale=0.55]
\draw[thick] (0,0) -- (0,5) -- (10,5) -- (10,0);
\draw[color=red,very thick] (0,0) -- (20,0);
\draw[color=green,very thick] (10,5) -- (20,0);
\draw (0,0) node[circle,fill=black!10,draw] {$1$};
\draw (0,5) node[circle,fill=black!10,draw] {$2$};
\draw (10,0) node[circle,fill=black!10,draw] {$3$};
\draw (10,5) node[circle,fill=black!10,draw] {$4$};
\draw (20,0) node[circle,fill=black!10,draw] {$5$};
\end{tikzpicture}
\end{center}
\caption{Illustration for link addition and deletion.}
\end{figure}

\begin{example}
	With these notational conventions we are now equipped to address link formation processes. To illustrate this notation, consider the network $g = \{ 12,13,24,34,35 \}$ on $N = \{ 1,2,3,4,5 \}$ as depicted in Figure 1 above consisting of the red and black links. Considering the green link $45 \notin g$, then $g' = g + 45 = \{ 12,13,24,34,35,45 \}$ is depicted in Figure 1 as the network consisting of all coloured links. Finally, removing the red link set $h = \{ 13,35 \} \subset g$ from $g$ results into $g'' = \{ 12,24,34 \}$, depicted by collection of the black links only in Figure 1.
\end{example}

\paragraph{Payoffs.}

Throughout the literature on game theoretic approaches to network formation, players are assumed to be fully incentivised in their drive to build and maintain links as well as delete links in existing networks. These incentives are introduced as a individualised payoff function. Indeed, for every player $i \in N$ we introduce \emph{player $i$'s network payoff function} as $\varphi_i \colon \mathbb{G}^N \to \mathbb{R}$, which assigns to every network $g \in \mathbb{G}^N$ a value $\varphi_i (g)$ that evaluates $i$'s situation as a member of the networked community described by $g$.

We can now capture all payoff information on the population $N$ of players in the \emph{network payoff function} given by $\varphi = ( \varphi_1, \ldots , \varphi_n ) \colon \mathbb{G}^N \to \mathbb{R}^N$. In particular, I emphasise that the function $\varphi$ indeed captures all incentives for the decision makers in $N$ in the network formation processes to be considered next.\footnote{We might refer to the multi-dimensional function $\varphi$ also as representing the \emph{network payoff structure}.}

A network payoff for a player captures \emph{all} values emanating in the structured community that is perceived or received by that player. This includes all perceived \emph{externalities} of third parties. In this regard, the network payoff function can capture widespread externalities from relationship and network formation in that community. The addition of network externalities in the payoff structure differentiates this inclusive network payoff approach from the more classical cooperative game theoretic payoff structure employed by \citet{Myerson1977,Myerson1980}, \citet{DuttaMutuswami1997} and \citet{Nouweland1993,Nouweland2004review}. The payoff function including widespread externalities has been seminally introduced in network theory by \citet{JacksonWolinsky1996}.
\begin{example}
	I illustrate this concept by revisiting the networks depicted in Figure 1. For example, player 1 can be assigned $\varphi_1 (g) =1$ as well as $\varphi_1 (g') = 5$ even though her neighbours in both networks are exactly the same, i.e., $N_1 (g) = N_1 (g') = \{ 2,3 \}$. This, therefore, captures widespread externalities from the creation of the link $45$ in the network $g$ from the perspective of player 1.
\end{example}

\subsection{Myerson's approach to network formation}

The most fundamental and basic model of how networks form under mutual consent was seminally introduced as an example in \citet[page 448]{Myerson1991}. He pointed out that in a very simple network formation game---known as the \emph{Myerson model}---, the resulting networks that are supported by Nash equilibria in this game \emph{always} include the empty network $g^0$. Hence, building no links at all is an equilibrium in the incentive structure generated by player benefits to network formation.

Myerson presented this as a negative insight, since it indicates that purely noncooperative game theory cannot provide a fertile basis for a debate of how non-trivial networks between players emerge. However, what this really expresses is that networks are not forming if players act purely selfishly. My contention is throughout that it actually has to be expected that pure selfishness would undermine cooperative acts such as forming links between pairs of players.

Here I initially explore the seminal Myerson model itself. In subsequent sections I turn to extensions of this basic model with added consideration of link formation costs. For the proper development of the Myerson model we need to review some basic non-cooperative game theory.

\paragraph{Preliminaries: Some game theory.}

This section relies heavily on standard noncooperative game theory. Again we let $N = \{ 1, \ldots , n \}$ be the set of players. A \emph{game} on $N$ is a pair $(\mathcal{A}, \pi )$ with $\mathcal{A} = ( A_1, \ldots , A_n )$ an ordered collection of \emph{strategy sets} such that each player $i \in N$ is assigned her individual strategy set $A_i$ and a game theoretic \emph{payoff function} $\pi = (\pi_1 , \ldots , \pi_n) \colon A \to \mathbb R^N$ where $A = \prod_{i \in N} A_i$ is the set of all \emph{strategy tuples} generated in $\mathcal{A}$.

Hence, in a non-cooperative game, each player $i \in N$ is endowed with her individual strategy set $A_i$ and a payoff function $\pi_i \colon A \to \mathbb R$. The fundamental idea is that every player selects a strategy that optimises her payoffs, provided that other players also select strategies that affect this payoff. As such, a game is a mathematical representation of a social interaction situation. Game theory is now a collection of rules and tools that model how players make decisions in the context of such social interaction situations.

A strategy tuple is a list $a = (a_1, \ldots , a_n) \in A$. We use the convention that the list of strategies of players other than $i \in N$ are indicated by $a_{-i} = (a_1 , \ldots , a_{i-1} , a_{i+1}, \ldots , a_n) \in \prod_{j \in N \colon j \neq i} A_j$. Hence, $a = (a_i, a_{-i})$.
\begin{definition}
	A strategy tuple $a^* \in A$ is a \textbf{Nash equilibrium} in the game $(\mathcal A, \pi )$ if for every player $i \in N$ and any strategy $b_i \in A_i$ it holds that $\pi_I (a^*) \geqslant \pi_i (b_i,a^*_{-i} )$. In a Nash equilibrium, every player optimises her strategy, \emph{given the strategic choices of all other players}.
\end{definition}
A Nash equilibrium can also be expressed in terms of ``best responses''. Formally, a strategy $a_i \in A_i$ is a \emph{best response} to strategy tuple $a_{-i} \in \prod_{j \in N \colon j \neq i} A_j$ if for every strategy $a'_i \in A_i$ it holds that $\pi_i (a_i , a_{-i}) \geqslant \pi_i (a'_i , a_{-i} )$. Hence, a best response is the strategy for player $i$ that optimises her payoffs given that all other players $j \neq i$ select the strategy $a_j \in A_j$.

Now a strategy tuple $a^* \in A$ is a Nash equilibrium if and only if for every player $i \in N$ it holds that $a^*_i$ is a best response to $a^*_{-i}$. As such a Nash equilibrium is a fixed point of the best response correspondence that is generated by the game. Furthermore, it can be shown that in this respect a Nash equilibrium usually can be interpreted as a saddle point in a well-constructed geometric representation of the game.

\paragraph{The Myerson model.}

\citet{Myerson1991} introduced his approach to modelling the formation of networks as an illustration of the underlying processes that determine the Nash equilibria in a non-cooperative strategic form game. Myerson's framework is \emph{the} quintessential model of mutual consent in link formation. The Myerson model encompasses a basic signalling game in which players send each other messages about whether they want to form a link or not. Due to its very fundamental and basic nature, it is a model that acts as the benchmark in any discussion on consent in link formation.

In Myerson's framework, players costlessly signal to each other whether they want to form links. Now, a link is established if and only if the two players signal both that they would like to form the link. Formally, the \emph{Myerson model} $\Gamma^m_\varphi$ on player set $N$ under network payoff function $\varphi \colon \mathbb{G}^N \to \mathbb{R}^N$ is a non-cooperative game $\Gamma^m_\varphi = (\mathcal A^m,\pi^m)$ given as follows:
\begin{itemize}
\item For every player $i \in N$, her strategy set is given by all vectors of signals to other players in $N \colon$
\begin{equation}
A^m_i = \left\{ \ell_i = (\ell_{i1} , \ell_{i2}, \ldots , \ell_{in} ) \, \left| \, \ell_{ij} \in \{ 0,1 \} \mbox{ and } \ell_{ii}=1 \right. \right\} \ ;
\end{equation}
Here, $\ell_{ij}$ is a signal that player $i$ communicates to player $j$ about her intentions to form a link with $j$. If $\ell_{ij} =1$, player $i$ indicates that she is interested in forming the link with player $j$; if $\ell_{ij} =0$, player $i$ signals that she wants to remain unattached to player $j$.

\item A link $ij$ is now formed if both players $i$ and $j$ signal to each other they want to form the link, i.e., if $\ell_{ij} = \ell_{ji} =1$. If we denote by $\ell = (\ell_1, \ldots , \ell_n ) \in A^m = A^m_1 \times \cdots \times A^m_n$ a strategy profile, then the resulting network can be identified as
\begin{equation}
g (\ell ) = \{ ij \in g^N \mid \ell_{ij} = \ell_{ji} =1 \} .
\end{equation}
We say that $g (\ell )$ is the network \emph{supported} by the strategy profile $\ell$ in the Myerson model.

\item The Myerson model is completed by the game theoretic payoff function $\pi^m \colon A^m \to \mathbb{R}^N$ defined by
\begin{equation}
\pi^m_i (\ell ) = \varphi_i \left( g( \ell ) \right) .
\end{equation}
Clearly, the payoff function $\pi^m$ reflects the property that signalling is costless and that there are no costs incurred in the formation of a link between any pair of players.
\end{itemize}
In the next discussion, I investigate the networks that are supported through Nash equilibria in the Myerson model.

\paragraph{M-networks.}

The Nash equilibria in the basic Myerson model form a class of signalling profiles that support networks on $N$ that are stable against unilateral modification. We denote these Nash equilibrium networks as ``M-networks'' to distinguish this class of networks from other classes of networks.
\begin{definition}
	Let $\varphi$ be a network payoff function on player set $N$ and let $\Gamma^m_\varphi = (\mathcal A^m,\pi^m)$ be the corresponding basic Myerson model. A network $g \in \mathbb{G}^N$ is an \textbf{M-network} if there exists a Nash equilibrium strategy tuple $\ell^g \in A^m$ in $\Gamma^m_\varphi$ such that $g (\ell^g) =g$.
\end{definition}
Clearly, using the Nash equilibrium conditions and the definition of $\pi^m$, we get the following M-network requirement: For every player $i \in N$ and every signal vector $\ell_i \in A^m_i$ it holds that
\[
\varphi_i \left( \, g ( \ell_i , \ell^g_{-i} ) \, \right) \leqslant \varphi_i \left( \, g ( \ell ) \, \right) .
\]
The concept of M-network is at the core of the assessment of network formation itself, since it describes the stable outcomes of the basic signalling framework represented in the Myerson model. Crucially, \citet{Myerson1991} already pointed out that the empty network is always supported as an M-network. Formally, this can be expressed as follows.
\begin{proposition}[Myerson's Lemma]
In the Myerson model $\Gamma^m_\varphi = (\mathcal A^m,\pi^m)$ the ``no-link'' signal profile $\ell^0 = (0, \ldots ,0) \in A^m$ is a Nash equilibrium. Consequently, the empty network $g^0 = g (\ell^0)$ is an M-network.
\end{proposition}
\begin{proof}
Let $\ell^0_{ij} =0$ for all $i,j \in N$, making up the strategy profile $\ell^0$. Then, for any player $i \in N$, any signal vector $\ell_i \in A^m_i$ is a best response to $\ell^0_{-i}$, since $g (\ell^0_{-i} , \ell_i )= g^0$ irrespective of the selected signal vector $\ell_i$. Therefore, $\ell^0_i$ itself is a best response to $\ell^0_{-i}$, showing that $\ell^0$ is a Nash equilibrium in $\Gamma^m_\varphi = (\mathcal A^m,\pi^m)$.
\end{proof}

\medskip\noindent
This property points out that non-trivial M-networks are very hard to form; rational self-interest easily results in complete failure and no cooperation might emerge. In this case, Myerson's Lemma indicates that, without some supporting mechanism, there simply are no incentives to justify that any links are formed at all. So, Myerson's Lemma points to the very fundamental issue of human cooperation: Why would rational human beings be cooperative? In this regard, Myerson's Lemma is a very succinct expression of this major question in social science and economics.

\section{Jackson-Wolinsky stability concepts}

The challenge of modelling non-trivial network formation stated in the discussion of the Myerson model as Myerson's Lemma was taken on by \citet{JacksonWolinsky1996}.  They formulated \emph{cooperative} equilibrium concepts that are tailored to the specific demands of modelling bilateral link formation.  This resulted in the notion of a ``pairwise stable'' network.

 I first discuss a class of cooperative or pairwise concepts of network stability from a link-based perspective as explored in \citet{Consent2006,GillesChakrabarti2012}. This concerns four fundamental link-stability principles, each founding a particular form of cooperative stability, and three further derived stability notions---including the seminal pairwise stability concept introduced by \citet{JacksonWolinsky1996}. 

Central to this approach is that while mutual consent is required for establishing a link, a player is able to delete her links unilaterally. Here, we focus on link-centred considerations. Hence, how would the deletion of one or more links affects the players' payoffs? Similarly, how would the addition of one or more links affect payoffs? These mutual considerations are brought together into a link- or network-based notion of stability.

\paragraph{Deleting links from networks.}

Throughout it is assumed that players have full autonomy or sovereignty over the decision to delete one or more of her links. Indeed, the principle of mutual consent requires that players control which links they participate in. This implies that every player can veto her participation in any link or relationship. Based on this consideration, I introduce two fundamental stability concepts concerning the deletion of links.

As before, let $\varphi \colon \mathbb G^N \to \mathbb R^N$ be a network payoff function on the player set $N$.
\begin{numm}
\item A network $g\in \mathbb{G}^{N}$ is \textbf{link deletion proof} (LDP) for $\varphi$ if for every player $i\in N$ and every neighbour $j \in N_{i} (g)$, it holds that $\varphi_{i}(g-ij)\leqslant \varphi_{i}(g)$.
\\
Link deletion proofness requires that no player has an incentive to sever an existing link with one of her neighbours.
\\
We denote by $\mathcal{D} (\varphi ) \subset \mathbb{G}^N$ the class of all link deletion proof networks for the given payoff function $\varphi$ \citep{JacksonWolinsky1996}.

\item A network $g\in \mathbb{G}^{N}$ is \textbf{strong link deletion proof} (SLDP) for $\varphi$ if for every player $i\in N$ and every set of her direct links $h \subset L_{i}(g)$, it holds that $\varphi_{i} (g -h) \leqslant \varphi_{i}(g)$.
\\
Strong link deletion proofness requires that no player has incentives to sever links with one or more of her neighbours simultaneously.
\\
We denote by $\mathcal{D}_s (\varphi ) \subset \mathbb{G}^N$ the class of all strong link deletion proof networks for the given payoff function $\varphi$ \citep{Consent2006}.
\end{numm}
From the definition it is clear that any SLDP network is always LDP and, therefore, strong link deletion proofness is indeed a stronger notion than (regular) link deletion proofness. As indicated, LDP was seminally introduced in \citet{JacksonWolinsky1996}, while SLDP was only introduced as a stand-alone concept in early drafts of  \citet{Consent2006}.

Second, the empty network $g^0 = \varnothing$ on any set of players $N$ is trivially strong link deletion proof. Indeed, this network does not contain any links and, therefore, the deletion of links is vacuously satisfied. We can therefore summarise that:
\begin{proposition}
	For any network payoff function $\varphi \colon \mathbb{G}^N \to \mathbb{R}^N$ it holds that
	\begin{equation}
		g^0 \in \mathcal{D}_s (\varphi ) \subset \mathcal{D} (\varphi) \subset \mathbb{G}^N.
	\end{equation}
\end{proposition}
The first question that I consider is under which conditions link deletion proofness is exactly the same as strong link deletion proofness. This seems a rather innocuous question, since SLDP is so much stronger a concept than LDP. Nevertheless, it is enlightening to identify the exact property on the network payoff structure $\varphi$ that allows this equivalence.
\begin{theorem} \label{9:equiv:Deletion}
Strong link deletion proofness and link deletion proofness are equivalent for network payoff structure $\varphi$ in the sense that $\mathcal{D} (\varphi ) = \mathcal{D}_s (\varphi )$ if and only if the network payoff structure $\varphi$ is \textbf{convex} on the class of link deletion proof networks $\mathcal{D} (\varphi ) \subset G^N$ in the sense that for every LDP network $g \in \mathcal{D} (\varphi )$, every player $i \in N$, every neighbour $j \in N_i (g)$ and every link set $h \subset L_i$ with $h \cap L_i (g) = \varnothing$ it holds that
\begin{equation}
\sum_{ij \in h} \left[ \, \varphi_i (g + ij) - \varphi_i (g) \, \right] \geqslant 0 \mbox{ implies that } \varphi_i (g+h) \geqslant \varphi_i (g) .
\end{equation}
\end{theorem}
For a proof of Theorem \ref{9:equiv:Deletion} I refer to Appendix A.1 of this survey.

The convexity property on the payoff structure $\varphi$ requires that the sign of the sum of values from adding one link to a network from a set of links fully determines whether adding all links is beneficial or not. Hence, looking at links one-by-one gives complete information about whether it is beneficial to add all links to the network or not.

\paragraph{Adding links to networks.}

Next I consider how players assess the addition of a link to an existing network. Again we take the idea of consent in link formation as central into our reasoning here. This implies that both parties in the formation of a new link have to agree that adding this link is beneficial.
\begin{numm}
\setcounter{llst}{2}
\item  A network $g\in \mathbb{G}^{N}$ is \textbf{link addition proof} (LAP) for $\varphi$ if for all $i,j\in N$ with $ij \notin g$, it holds that $\varphi_{i} (g+ij) > \varphi_{i}(g) $ implies $\varphi_{j} (g+ij)< \varphi_{j}(g)$.
\\
Link addition proofness states that there are no incentives for any pair of players to form an additional link. This is based on the requirement of mutual consent in link formation. Indeed, if one player would like to add a link, the other player would have strong objections. In this case this is formulated as that, if one player has benefits from forming the link, the other (consenting) party has losses and, thus, would withhold her consent.
\\
We denote by $\mathcal{A} (\varphi ) \subset \mathbb{G}^N$ the class of all link addition proof networks for the given payoff function $\varphi$ \citep{JacksonWolinsky1996}.

\item A network $g\in \mathbb{G}^{N}$ is \textbf{strict link addition proof} (SLAP) for $\varphi$ if for all $i,j\in N$, it holds that $ij \notin g$ if and only if $\varphi_{i} (g+ij) < \varphi_{i}(g) $ as well as $\varphi_{j} (g+ij) < \varphi_{j}(g)$.
\\
Strict link addition proofness is a far stronger notion that LAP. Indeed, it requires that both players agree that forming an additional link between them is not beneficial for either of them. This agreement is imposed and only a certain very specific type of network payoff structures would support such networks to exist. Consequently, it has to be expected that, for an arbitrary regular network payoff function, only a rather small class of networks actually satisfies this property.
\\
We denote by $\mathcal{A}_s (\varphi ) \subset \mathbb{G}^N$ the class of all strict link addition proof networks for the given payoff structure $\varphi$ \citep{GillesSarangi2010}.
\end{numm}
The introduced notions of link addition proofness require some clarification. These two notions indeed only seem to partially cover the idea that a network is stable if it satisfies the property that ``if $i$ has an incentive to form an additional link with $j$, then $j$ has no incentive to form a link with $i$''. This is subject to the next discussion.

To understand link addition proofness in more detail, we can reformulate it. Indeed, a network $g$ is link addition proof if and only if for all players $i,j \in N$ with $ij \notin g \colon$
\begin{equation}
\varphi_i (g + ij) \geqslant \varphi_i (g) \mbox{ implies } \varphi_j (g + ij) \leqslant \varphi_j (g) .
\end{equation}
This has some interesting consequences regarding the interpretation of the LAP property. First, a link $ij \notin g$ for some $i,j \in N$ is \emph{non-discerning} if it holds that
\begin{equation}
	\varphi_i (g + ij) = \varphi_i (g) \mbox{ as well as } \varphi_j (g + ij) = \varphi_j (g).
\end{equation}
From the derivation above, the definition of link addition proofness is indeed ambiguous whether any non-discerning link $ij$ should be in the network for it to be LAP or not. Hence, such non-discerning links can arbitrarily be added to or deleted from networks without the LAP property being affected. Thus, the class of non-discerning links makes the determination of LAP networks ``fuzzy''.

To address this issue of the addition or deletion of non-discerning links, I introduce a third type of link addition proofness:
\begin{numm}
\setcounter{llst}{4}
\item A network $g\in \mathbb{G}^{N}$ is \textbf{$\star$-link addition proof} ($\star$-LAP) for $\varphi$ if for all players $i,j\in N$, it holds that if $ij \notin g$, then $\varphi_{i} (g+ij) \geqslant \varphi_{i}(g) $ implies $\varphi_{j} (g+ij)< \varphi_{j}(g)$.
\\
We denote by $\mathcal{A}_{\star} (\varphi ) \subset \mathbb{G}^N$ the class of all $\star$-link addition proof networks for the given payoff structure $\varphi$.
\end{numm}
This minor modification of the definition of link addition proofness simply requires that all non-discerning links should be part of a $\star$-link addition proof network. This makes the definition unambiguous.
\begin{example}
To delineate the three link addition proofness concepts introduced here, we can explore an example of a network payoff function in which these concepts result in different classes of networks. We consider three players and all possible networks, i.e., $N = \{ 1,2,3 \}$ and $\mathbb G^N = \{ g \mid g \subset g^N \}$ where $g^N = \{ 12,23,13 \}$. Note that there are exactly eight possible networks on $N$, i.e., $\# \mathbb G^N =8$.
\\
We now consider a particular network payoff function $\varphi$ on the generated class of networks $\mathbb{G}^N$ on $N$. All potential network payoffs represented by $\varphi$ can be represented in an appropriately constructed table:

\begin{center}
\begin{tabular}{|l|c|c|c||c|}
\hline
\textbf{Network} $g$ & $\varphi_1 (g)$ & $\varphi_2 (g)$ & $\varphi_3 (g)$ & \textbf{Stability} \\
\hline
$g^0 = \varnothing$ & 0 & 0 & 0 & LAP \\
$g^1 = \{ 12 \}$ & 0 & 0 & 1 & $\star$-LAP \\
$g^2 = \{ 13 \}$ & 0 & 0 & 0 & \\
$g^3 = \{ 23 \}$ & 0 & 0 & 0 & \\
$g^4 = \{ 12,13 \}$ & 2 & 1 & 0 & \\
$g^5 = \{ 12,23 \}$ & 1 & 2 & 0 & \\
$g^6 = \{ 13,23 \}$ & 0 & 1 & 0 & \\
$g^7 = g^N$ & 3 & 3 & 3 & SLAP \\
\hline
\end{tabular}

\end{center}
First, note that $g^0$ is link addition proof, but not $\star$-link addition proof. Indeed, if any link is added to the empty network, no payoffs are changed for any of the players involved. On the other hand, there are no losses, thus precluding that $g^0$ is $\star$-link addition proof.
\\
Next, $g^1$ is $\star$-link addition proof, but not strong link addition proof. Indeed, any addition of a link to $g^1$ results into a loss for player 3. However, adding link $13$ results into a strict gain for player 1, implying that $g^1$ is not strong link addition proof.
\\
Third, the complete network $g^N$ is strong link addition proof by tautology. Indeed, there are no links to be added to this network, and therefore vacuously the property of strong link addition proofness is satisfied.
\\
I remark that none of the other networks have any link addition proofness properties.
\end{example}
Next I explore the equivalence of these link addition proofness concepts. In order to explore these equivalences effectively, I introduce two auxiliary properties of the network payoff structure.
\begin{definition}
Consider a network payoff structure $\varphi$ on $\mathbb{G}^N$. Then:
\begin{itemize}
\item The structure $\varphi$ is said to be \textbf{discerning} on the class of networks $\mathbb{G} \subset \mathbb{G}^N$ if for every network $g \in \mathbb{G}$ it holds that for any pair $i,j \in N$ with $ij \notin g$ either $\varphi_{i} (g+ij) \neq \varphi_{i}(g)$ or $\varphi_{j} (g+ij) \neq \varphi_{j}(g)$ or both.

\item The structure $\varphi$ is said to be \textbf{uniform} on the class of networks $\mathbb{G} \subset \mathbb{G}^N$ if for every network $g \in \mathbb{G}$ and for any pair $i,j \in N$ with $ij \notin g$ it holds that
\begin{equation}
\varphi_i (g + ij) \geqslant \varphi_i (g) \mbox{ implies } \varphi_j (g + ij) \geqslant \varphi_j (g) .
\end{equation}
\end{itemize}
\end{definition}
Using these auxiliary concepts we can now show the following equivalences:
\begin{theorem} \label{9:equiv:Addition}
Let $\varphi$ be some network payoff structure on the class of all networks $\mathbb{G}^N$ on the set of players $N$. Then the following properties hold:
\begin{abet}
\item $g^N \in \mathcal{A}_s (\varphi ) \subset \mathcal{A}_{\star} (\varphi ) \subset \mathcal{A} (\varphi )$;

\item It holds that $\mathcal{A}_{\star} (\varphi ) = \mathcal{A} (\varphi )$ if and only if $\varphi$ is discerning on $\mathcal{A} (\varphi )$, and;

\item It holds that $\mathcal{A}_s (\varphi ) = \mathcal{A}_{\star} (\varphi )$ if and only if $\varphi$ is uniform on $\mathcal{A}_{\star} (\varphi )$.
\end{abet}
\end{theorem}
For a proof of Theorem \ref{9:equiv:Addition} I refer to Appendix A.2 in this survey. Furthermore, from Theorem \ref{9:equiv:Addition} it is easily concluded that the following equivalence also holds:
\begin{corollary}
SLAP and LAP are equivalent concepts for payoff structure $\varphi$ in the sense that $\mathcal{A}_s (\varphi ) = \mathcal{A} (\varphi )$ if and only if the payoff structure $\varphi$ is discerning and uniform on $\mathcal{A} (\varphi )$.
\end{corollary}

\subsection{Notions of pairwise stability}

In the previous discussion, I introduced four fundamental stability concepts on adding links to and deleting links from a network. These four basic notions can be combined to define derived concepts. The first concept---known as \emph{pairwise stability} \citep{JacksonWolinsky1996}---combines the weakest link stability notions and has been the subject of extensive discussion in the literature. This notion implicitly assumes that players only consider the deletion and addition of one specific link at a time.
\begin{numm}
\setcounter{llst}{5}
\item Network $g$ is \textbf{pairwise stable} (PS) for $\varphi$ if $g$ is link deletion proof as well as link addition proof. We denote by $\mathcal{P} (\varphi ) = \mathcal{D} (\varphi ) \cap \mathcal{A} (\varphi )$ the family of pairwise stable networks for the payoff function $\varphi $.
\end{numm}
The original pairwise stability concept---introduced by \citet{JacksonWolinsky1996}---only concerns itself with the contemplation of adding a single link to or deleting a single link from a given network. If there are no incentives for players to either add a link to the existing network or delete a link from the network, then the network is ``pairwise stable'': There are no incentives present under the hypothesis of mutual consent in link formation that anybody wants to change a single link in this network.

Two further derived stability concepts, which strengthen the notion of pairwise stability, have particular relevance in the theory of consent in link formation. Strong pairwise stability \citep{Consent2006,GillesChakrabarti2012} assumes that players can delete an arbitrary collection of links under their control. Hence, they can veto any link in which they participate. On the other hand, the contemplation of adding links remains confined to adding a single link.

Strict pairwise stability \citep{GillesSarangi2010} is the strongest notion in this framework. It not only considers that players can delete any number of their existing links, but also that they are assumed to be in agreement regarding the addition of a link to an existing network. It is clear that for an arbitrary network payoff structure, the collection of such strictly pairwise stable networks might well be empty. Only for certain network payoff structures such networks might emerge.
\begin{numm}
\setcounter{llst}{6}
\item Network $g$ is \textbf{strongly pairwise stable} (SPS) for $\varphi$ if it is strong link deletion proof as well as link addition proof.
\\
We denote by $\mathcal{P}_{\star} (\varphi ) = \mathcal{D}_s (\varphi ) \cap \mathcal{A} (\varphi )$ the family of strongly pairwise stable networks for the payoff function $\varphi $.

\item Network $g$ is \textbf{strictly pairwise stable} (SPS*) for $\varphi$ if it is strong link deletion proof as well as strict link addition proof.
\\
We denote by $\mathcal{P}_{s} (\varphi ) = \mathcal{D}_s (\varphi ) \cap \mathcal{A}_s (\varphi )$ the family of strictly pairwise stable networks for the payoff function $\varphi $.
\end{numm}
These three pairwise stability concepts generate different classes of networks in most cases. I consider an example to illustrate this.
\begin{example}
Again consider three players and all potentially generated networks, i.e., $N = \{ 1,2,3 \}$ with $g^N = \{ 12,23,13 \}$. Now, consider a network payoff function $\varphi$ on the generated class of networks $\mathbb{G}^N$ on $N$ represented in the following table:

\smallskip
\begin{center}
\begin{tabular}{|l|c|c|c||c|}
\hline
\textbf{Network} $g$ & $\varphi_1 (g)$ & $\varphi_2 (g)$ & $\varphi_3 (g)$ & \textbf{Stability} \\
\hline
$g^0 = \varnothing$ & 0 & 0 & 0 & Strongly PS \\
$g^1 = \{ 12 \}$ & 0 & 0 & 5 & Strictly PS \\
$g^2 = \{ 13 \}$ & 0 & 0 & 0 & \\
$g^3 = \{ 23 \}$ & 0 & 0 & 0 & \\
$g^4 = \{ 12,13 \}$ & -1 & 0 & 0 & \\
$g^5 = \{ 12,23 \}$ & 0 & -1 & 0 & \\
$g^6 = \{ 13,23 \}$ & 0 & 1 & 1 & \\
$g^7 = g^N$ & 3 & 3 & 3 & PS \\
\hline
\end{tabular}

\end{center}
Again we discuss the properties of these networks.
\\
First, note that the empty network $g^0$ is trivially SLDP and in this case as well LAP. Therefore, it is indeed strongly pairwise stable.\footnote{It should be remarked that networks with at most one link are SLDP if they are LDP. Therefore, they are strongly pairwise stable if they are link addition proof and link deletion proof.}
\\
Second, $g^1$ is LDP and, therefore, SLDP. Moreover, $g^1$ is SLAP. Indeed, adding link $13$ to $g^1$ results into strict losses for both players 1 and 3. Similarly, for link $23$. Thus, we conclude that $g^1$ is strictly pairwise stable.
\\
Finally, the complete network $g^N$ is SLAP due to being the maximal network. Furthermore, $g^N$ is LDP. However, $g^N$ is not SLDP. player 3 has the strict incentive to delete both her links and revert to network $g^1$.
\\
We conclude from this discussion that this simple network payoff example induces three distinct classes of pairwise stable networks.
\end{example}
Using the equivalence results stated in Theorems \ref{9:equiv:Deletion} and \ref{9:equiv:Addition}, we can now conclude the following equivalences between the formulated pairwise stability concepts. The proofs are rather transparent and therefore omitted.
\begin{corollary}
Consider a network payoff structure $\varphi$ on the class of all networks $\mathbb{G}^N$ on set of players $N$. Then the following relationships hold:
\begin{abet}
\item $\mathcal{P}_s (\varphi ) \subset \mathcal{P}_{\star} (\varphi ) \subset \mathcal{P} (\varphi )$;

\item Pairwise stability and strong pairwise stability are equivalent concepts for $\varphi$ in the sense that $\mathcal{P} (\varphi ) = \mathcal{P}_{\star} (\varphi )$ if and only if $\varphi$ is convex on $\mathcal{P} (\varphi )$;

\item Strong pairwise stability and strict pairwise stability are equivalent concepts for $\varphi$ in the sense that $\mathcal{P}_{\star} (\varphi ) = \mathcal{P}_{s} (\varphi )$ if and only if $\varphi$ is discerning and uniform on $\mathcal{P}_{\star} (\varphi )$, and;

\item Pairwise stability and strict pairwise stability are equivalent concepts for $\varphi$ in the sense that $\mathcal{P} (\varphi ) = \mathcal{P}_{s} (\varphi )$ if and only if $\varphi$ is convex, discerning as well as uniform on $\mathcal{P} (\varphi )$.
\end{abet}
\end{corollary}

\subsection{Strong stability}

Next I discuss some of the ideas put forward by \citet{JacksonNouweland2005}. They investigated networks that emerge if coalitions of arbitrary size can make changes to the network in a coordinated fashion to the coalition's overall benefit.\footnote{This approach is akin to the strong equilibrium concept proposed by \citet{Aumann1959} in non-cooperative game theory. Jackson-Nouweland's concept of strong stability can be viewed as a network theoretical implementation of the ideas behind Aumann's strong equilibrium concept.} As such strong stability is an extension of the pairwise stability concept to allow arbitrary coalitions to adjust the network structure under their control.

As a preliminary we denote a \emph{coalition} as any subset $S$ of players in $N$; hence, a coalition is any $S \subset N$. This includes the empty coalition $\varnothing$ as well as the ``grand'' coalition $N$ itself. In a non-cooperative game $( \mathcal A, \pi )$, for any coalition $S \subset N$ and strategy profile $a \in A$ we denote by $a_S$ the $S$-restriction of $a$ defined by $(a_j)_{j \in S}$ and by $a_{N \setminus S}$ its complement $(a_k)_{k \notin S}$.

Now, in a non-cooperative game $( \mathcal A, \pi )$ a strategy tuple $a \in A$ is a \emph{strong equilibrium} if for every (non-empty) coalition of players $S \subset N$ and every coordinated strategic deviation $b_S = (b_i)_{i \in S} \in A_S = \prod_{i \in S} A_i$ it holds that
\begin{equation}
	\pi_i \left( a_{N \setminus S} , b_S \right) \leqslant \pi_i (a) \qquad \mbox{ for all } i \in S
\end{equation}
Next we introduce the strong stability concept put forward by \citet{JacksonNouweland2005}. The next definition essentially transposes strong equilibrium conditions to network formation situations.
\begin{definition}
Let $\varphi$ be a network payoff function on $N$ and consider the corresponding Myerson model $\Gamma^m_\varphi = (\mathcal A^m,\pi^m)$.
\begin{numm}
	\item A network $g' \in \mathbb{G}^N$ \textbf{can be obtained from} network $g \in \mathbb{G}^N$ through the coordinated actions of coalition $S \subset N$ if $g' = g +h^+ - h^-$, where $h^+ \subset g^S = \{ ij \mid i,j \in S \}$ and $h^- \subset \cup_{i \in S} \, L_i (g)$.
	\item A network $g \in G^N$ is \textbf{strongly stable} if for every coalition $S \subset N$ and every network $g'$ that is obtainable from network $g$ through coordinated actions from coalition $S$ it holds that $\varphi_i (g') > \varphi_i (g)$ for some player $i \in S$ implies that there exists some other player $j \in S$ with $\varphi_j (g') < \varphi_j (g)$.
\end{numm}
\end{definition}
It should be remarked that \citet{DuttaMutuswami1997} introduced a slightly different definition of ``strong stability''. They consider that all members of $S$ need to be made strictly better off for a deviation to be successful.\footnote{In the definition used by \citet{JacksonNouweland2005} a deviation needs to make all members of $S$ to be at least as well off and making one member strictly better off. }

Strong equilibrium is a very demanding concept and these equilibria do not exist in many game theoretic decision situations. Similarly, the notion of strong stability is equally demanding, resulting that such networks rather unlikely exist. The next example illustrates these issues and introduces the notion of \emph{costly} link formation that will be explored further in the next two subsections.
\begin{example} \label{ex:StrongStab} \textbf{(Costly trade networks)} \\
This example of a Walrasian trade network has been introduced seminally in \citet{JacksonWatts2002} and further developed in \citet{JacksonNouweland2005} and \citet{GCS2011}. It considers an economy of $n$ players who trade goods through connection paths. There are two commodities $X$ and $Y$ and all players are endowed with a Cobb-Douglas utility function $u(x,y) = \sqrt{xy}$. All players are assumed to have a commodity endowment of either $(1,0)$ or $(0,1)$ with an equal probability of $\tfrac{1}{2}$.
\\
Players can trade with any other player that they are connected with, directly or indirectly. Hence, there emerge complete markets in each of the components. So, for $n=5$ a network $g = \{ 12,23,45 \}$ generates two components and two markets, namely 123 separated from 45. Additional links, therefore, not always contribute to the extent of these markets: $g' = \{ 12,23,13,45 \}$ results in exactly the same markets 123 and 45.
\\
The cost $c$ of forming any link $ij$ is uniform and set at $c > \tfrac{1}{2}$. The costs of the formation of the trade network are divided equally among the members of a market, being a component of the network.
\\
The network payoff function $\varphi$ is now defined as the expected net benefits from participating in the generated market structure. This can be developed as follows.
\\[1ex]
First, consider the case of a market of the size two. There is a probability of $\tfrac{1}{2}$ that these two players have opposite endowments and a probability of $\tfrac{1}{2}$ that they have the same endowment. Hence, the probability of trade is $\tfrac{1}{2}$ resulting in a Walrasian allocation of $( \tfrac{1}{2}, \tfrac{1}{2})$ resulting in $\varphi = \tfrac{1}{2} \cdot \sqrt{ \tfrac{1}{4}} - \tfrac{1}{2} c = \tfrac{1}{4} - \tfrac{1}{2} c <0$.
\\
More generally consider a market (component) of $k$ players. The probability of $r$ players having endowment $(0,1)$ and $(k-r)$ players having endowment $(1,0)$ is now
\[
^kC_r \left( \tfrac{1}{2} \right)^k-r \cdot \left( \tfrac{1}{2} \right)^r = {}^kC_r \left( \tfrac{1}{2} \right)^k .
\]
The expected gross payoff from trade is now given by
\[
\tfrac{r}{2k} \cdot \left( \frac{k-r}{r} \right)^{\tfrac{1}{2}} + \tfrac{k-r}{2k} \cdot \left( \frac{r}{k-r} \right)^{\tfrac{1}{2}} = \frac{\sqrt{r(k-r)}}{k}
\]
Hence, taking into account that there are exactly $k-1$ links required to build a market for $k$ players, the resulting net payoff from this trade network is given by
\[
\varphi = \frac{1}{k \cdot 2^k} \left[ \sum^{k-1}_{r=1} {}^kC_r \sqrt{r(k-r)} \right] - \frac{(k-1)c}{k} .
\]
Turning to $n=k=3$ it can easily be computed that the net benefits to each player are given by
\[
\varphi = \frac{\sqrt{2}}{4} - \frac{2c}{3} > 0 \mbox{ for } \tfrac{1}{2} < c < \tfrac{3 \sqrt{2}}{8} .
\]
For $n=k=3$ and the given link formation cost range there are two pairwise stable networks, namely the connected network and the (inefficient) empty network. The empty network is bilaterally stable, since creating a single link between two players is not beneficial for the given link formation cost range. On the other hand, the empty network is not strongly stable. Indeed, if all three players coordinate they would create two links to make a beneficial market among them.
\\
This also shows that the connected component based on two links among the three players is strongly stable.
\end{example}
In this section I discussed the stability concept and its variants in the link-based cooperative framework as seminally set out by \citet{JacksonWolinsky1996}. It is clear that these concepts are rather limited in their scope, since they are link-based only. Individual and collective incentives are not truly taken into account. Indeed, considerations are founded on adding and deleting links; the players' incentives are assumed to coincide with the (marginal) benefits generated from these links rather than the individualised payoffs. Next, I return to Myerson's original non-cooperative framework founded on the direct benefits to players to the formation of links.

\section{Refinements of M-networks}

In this section I review stability and equilibrium concepts that refine the class of M-networks that emerges from the Myerson approach to non-cooperative network formation under mutual consent. This literature is founded on the insight that the class of M-networks is very large. This is subject of the next theorem, which states the equivalence of the class of M-networks with the set of strong link deletion proof networks.
\begin{theorem} \label{9:prop:Mlargeness}
Let $\varphi$ be a network payoff function on $N$ and consider the corresponding Myerson model $\Gamma^m_\varphi = (\mathcal A^m,\pi^m)$.
\begin{abet}
\item A network $g \in \mathbb G^N$ is an M-network for $\varphi$ if and only if $g$ is strong link deletion proof for $\varphi$.

\item Suppose that the network payoff structure $\varphi$ is \textbf{link monotone} in the sense that for every player $i \in N$, every network $g \in \mathbb{G}^N$ and every link $ij \notin L_i (g)$ it holds that $\varphi_i (g + ij) \geqslant \varphi_i (g)$. Then every network $g \in \mathbb{G}^N$ is supported as an M-network.
\end{abet}
\end{theorem}
For a proof of this theorem I refer to Appendix A.3.

The fundamental insights presented as Myerson's Lemma and Theorem \ref{9:prop:Mlargeness} have motivated economists and social scientists to look into ``refinements'' of the Nash equilibrium concept in the Myerson model. These refinement equilibrium concepts have been developed particularly for addressing link formation issues from the perspective of consent. These attempts can be divided into two classes.

First, the standard approach in game theoretic models of network formation is to strictly apply methodological individualistic perspectives. Thus, all motivations emanate from the player decision makers and are not considered to be external to the rational decision making process. This has resulted into a number of equilibrium concepts that simply assume that decision makers have a natural ability to cooperate if the incentives are in favour of such cooperation. Below I present the refinements considered by \citet{BlochJackson2006} and \citet{GCS2011,GillesChakrabarti2012}.

The second approach is to explicitly assume that decision makers are \emph{not} fully individualistic, but adhere to some institutional or trusting norms of behaviour. \citet{Buskens2005} and \citet{GillesSarangi2010} explicitly introduce a model of trusting behaviour through the introduction of a individualised belief or conjecture that other decision makers will form links if they benefit from that. Thus, the trust in network formation is internalised into the player decision makers; all such decision makers adhere to a well-defined norm of decision making that expresses trusting behaviour. This is fully developed in Section 5.

Similarly, certain equilibrium concepts in non-cooperative game theory are founded on institutional signalling systems. The main such concept is Aumann's \emph{correlated equilibrium}, which can be used to introduce institutional arrangements in the decision making processes of players \citep{Aumann1974}. Here these institutions are explicitly modelled as external to these players. They adhere to these institutions since they benefit from applying these institutional behavioural rules instead of acting purely selfish. This is explored fully in Section 6.

\subsection{Pairwise Nash equilibrium and bilateral stability}

\citet{GoyalJoshi2006} introduced a refinement of the M-network concept that implements the idea of cooperation between players to modify the network through coordinated actions. Thus, it is assumed that decision makers can implement bilateral or pairwise coordinated network modification. So, we consider any pair of players $i,j \in N$ who consider how to modify their strategic signals $\ell_i$ and $\ell_j$ to modify the resulting network in their favour. 

This bilaterally coordinated action can be modelled in two different fashion. First, within the Myerson model as the so-called ``pairwise'' Nash equilibrium \citep{GoyalJoshi2006} and, second, as a network stability notion, denoted as ``bilateral'' stability \citep{GCS2011}.\footnote{I remark here that I use a terminology that deviates from the literature. Indeed, the pairwise Nash equilibrium concept in the Myerson model was seminally introduced in \citet{GoyalJoshi2006} and explored further by \citet{BlochJackson2006} and \citet{JoshiSarangi2020}. It refers to M-networks that are additionally link addition proof. Therefore, I use the notion of pairwise Nash equilibrium here in a slightly different way as introduced in \citet{GCS2011}.} This is introduced in the next definition.
\begin{definition}
Let $\varphi$ be a network payoff function on $N$ and consider the corresponding Myerson model $\Gamma^m_\varphi = (\mathcal A^m,\pi^m)$.
\begin{numm}
\item A signal profile $\ell \in A^m$ is a \textbf{pairwise Nash equilibrium} in $\Gamma^m_\varphi$ if $\ell$ is a Nash equilibrium in $\Gamma^m_\varphi$ and for every pair of players $i,j \in N$ it holds that
\begin{equation}
\pi^m_i \left( \ell'_i , \ell'_j , \ell_{-i,j} \right) > \pi^m_i (\ell) \mbox{ implies that } \pi^m_j \left( \ell'_i , \ell'_j , \ell_{-i,j} \right) < \pi^m_j (\ell)
\end{equation}
for all deviations $\ell'_i \in A^m_i$ and $\ell'_j \in A^m_j$. (Here, $\ell_{-i,j}$ refers to the restricted signal profile $(\ell_h)_{h \neq i,j}$.)

\item A network $g \in \mathbb{G}^N$ is \textbf{bilaterally stable} for $\varphi$ if $g$ is strong deletion proof for $\varphi$ and for every pair of players $i,j \in N$ and network $g' = g+\hat{h} -h_i - h_j$ with $\hat{h} \in \{ \, \{ ij \} , \varnothing \}$, $h_i \subset L_i (g)$ and $h_j \subset L_j(g)$ it holds that
\begin{equation}
\varphi_i (g') > \varphi_i (g) \mbox{ implies that } \varphi_j (g') < \varphi_j (g) .
\end{equation}
\end{numm}
\end{definition}
It is not hard to see that in the Myerson model there is a complete equivalence between these two concepts. The pairwise Nash equilibrium is simply a strategic formulation of bilateral stability. I give the following proposition therefore without proof.
\begin{proposition}
Let $\varphi$ be a network payoff function on $N$ and consider the corresponding Myerson model $\Gamma^m_\varphi = (\mathcal A^m,\pi^m)$. A network $g \in \mathbb{G}^N$ is supported through a pairwise Nash equilibrium $\ell \in A^m$ with $g(\ell )=g$ if and only if $g$ is bilaterally stable for $\varphi$.
\end{proposition}
Although these concepts are quite natural within the context of network formation, the additional benefits are rather limited. Coordinated pairwise activity is well captured by the three pairwise stability concepts that have been introduced in this survey. The notion of unilateral stability (See Section 5) also captures coordinated action in the sense that it is assumed that players respond positively to a player's proposal to change the network if that is to their benefit. Bilateral stability does not extend this to pairs of players, but reverts back to the normal best response rationality principle that others keep their actions unchanged.

\paragraph{Stability of higher orders.}

The notion of bilateral stability can easily be extended to stability of higher orders. Indeed, under bilateral stability it is assumed that coalitions of two players can modify the network as proposed above. This can be extended to coalitions of at most $r$ members, where $r \in \mathbb N$ is the assumed maximum size of the coalition under consideration. This is referred to as ``stability of order $r$'' in \citet{GCS2011}. In particular, if $r=n$, we arrive at the strong stability notion of \citet{JacksonNouweland2005}. This shows that these concepts represent intermediate stability notion between M-networks and strongly stable networks.\footnote{For results concerning these intermediate stability concepts, I refer to the quoted papers.}

\subsection{Two-sided link formation costs}

Example \ref{ex:StrongStab} introduced the idea that there are normally link formation costs. In this particular case the costs of network formation are borne equally among all players that participate in the network. This signifies a collective approach to the allocation of network formation costs. It is more natural to assume that players only bear the costs of the links that they participate in. Next, I develop the idea of link formation costs further and refine the notion of M-networks to capture this.

 In particular, I consider a modification of the Myerson model where the ``intent to form links'' is costly in the sense that approaching another player to form a link involves \emph{explicit} investment of time, effort and energy. Hence, the act of sending a signal is costly. However, if the other player does not reciprocate and the link does not materialise, the player choosing to ``reach out'' still incurs this cost.\footnote{This model of two-sided link formation costs was introduced in \citet{Consent2006} and developed further by \citet{GillesSarangi2010} and \citet{GillesChakrabarti2012}.} This means that if player $i \in N$ contemplates building a link $ij$ with player $j \in N$ and sends a message $\ell_{ij}=1$, she incurs a cost of $c_{ij} >0$. On the other hand, $\ell_{ij} =0$ signifies no link is attempted to be made, which imposes \emph{no} costs on player $i$.

Formally, a \emph{link formation cost structure} can therefore be represented by a function $c \colon N \times N \to \mathbb R_+$ where $c (i,j) = c_{ij}\geqslant 0$ is the cost that player $i \in N$ incurs for sending a message to player $j \in N$, using the convention that $c(i,i) =0$ for all $i \in N$. Hence, player $i$ incurs a cost $c_{ij} \geqslant 0$ when communicating to player $j$ that she wants to form a link. In particular, this cost refers to the effort to respond to messages sent by others. Obviously, if $c_{ij}=0$, then there is no cost to communicating and sending messages from $i$ to $j$.

This construction introduces the \emph{consent model with two-sided link formation costs} as a modification of the (basic) Myerson model $\Gamma^m_\varphi$ given as a non-cooperative game $\Gamma^a_\varphi (c) = ( \mathcal A^{a},\pi^{a})$, where player $i$'s strategy set is given by $A^a_i = A^m_i$ and player $i$'s payoff for any strategy tuple $\ell \in A^a$ is given by
\begin{equation}
\pi_{i}^{a}( \ell )=\varphi _{i}(g( \ell ))-\sum_{j\neq i} \ell_{ij}\cdot c_{ij} = \pi^m_i ( \ell ) -\sum_{j\neq i} \ell_{ij}\cdot c_{ij} ,
\end{equation}
where $\varphi \colon \mathbb{G}^N \to \mathbb{R}^N$ is the network payoff function representing the gross benefits from network formation without taking into account the costs of link formation.

Our first result develops a complete characterisation of the Nash equilibria in the consent model with two-sided link formation costs. Part of this equivalence theorem was already stated without proof in \citet{GillesSarangi2010} and as stated here is taken from \citet{GillesChakrabarti2012}. There are some preliminaries that need to be developed before stating the main assertion.
\begin{definition}
Let $\varphi$ be a network payoff function on player set $N$ and let $c \colon N \times N \to \mathbb R_+$ a link formation cost structure on $N$. Furthermore, let $\Gamma^a_\varphi (c) = ( \mathcal A^{a},\pi^{a})$ be the associated consent model with two-sided link formation costs.
\\
A strategy tuple $\ell \in A^a = A^m$ is \textbf{non-superfluous} in the consent model with two-sided link formation costs $\Gamma^a_\varphi = ( \mathcal A^{a},\pi ^{a})$ if for all pairs of players $i,j \in N$, $\ell_{ij}=1$ if and only if $\ell_{ji}=1$.
\\
We call a non-superfluous strategy tuple $\ell \in A^a$ that is a Nash equilibrium a \textbf{non-superfluous Nash equilibrium}.
\end{definition}
The main theorem states that in $\Gamma^a_\varphi (c)$ the networks that are supported by Nash equilibria are exactly the strong link deletion proof networks for a network payoff function that takes account of the link formation costs. For a proof of the next theorem I refer to Appendix A.4.
\begin{theorem} \label{9:thm:2-equiv}
Let $\varphi$ be a network payoff function on player set $N$ and let $c \colon N \times N \to \mathbb R_+$ be a link formation cost structure on $N$. Furthermore, let $\Gamma^a_\varphi (c) = ( \mathcal A^{a},\pi^{a})$ be the associated consent model with two-sided link formation costs.
\\
Then for every network $g \in \mathbb G^N$ the following three statements are equivalent:
\begin{abet}
\item Network $g$ is supported by a Nash equilibrium of the consent model with two-sided link formation costs $\Gamma^a_\varphi (c)$.

\item Network $g$ is supported by a non-superfluous Nash equilibrium of the consent model with two-sided link formation costs $\Gamma^a_\varphi (c)$.

\item Network $g$ is strong link deletion proof with regard to the network payoff function $\varphi^{a}\colon \mathbb{G}^{N}\rightarrow \mathbb{R}^{N}$ given by 
\begin{equation}
\varphi _{i}^{a}(g)=\varphi_{i}(g)-\underset{j\in N_{i}(g)}{\sum }c_{ij}
\end{equation}
\end{abet}
\end{theorem}
Theorem \ref{9:thm:2-equiv} provides a complete and detailed characterisation of the set of all Nash equilibria of the consent model with two-sided link formation costs. Furthermore, Theorem \ref{9:thm:2-equiv} clearly generalises the insight that the class of M-networks in the basic Myerson model is exactly the class of strong deletion proof networks under network payoff function $\varphi$.

In particular, each Nash equilibrium network is actually supported by a \emph{unique} non-superfluous strategy profile if the cost structure is non-trivial in the sense that all link formation costs are positive. \citet{GillesChakrabarti2012} also discuss that there actually exist superfluous Nash equilibria if costs of link formation are zero for one of the players.
\begin{example} \textbf{\citep{GillesChakrabarti2012}} \\
Consider the binary network formation situation with $N=\{ 1,2 \}$ and the network payoff function given by $\varphi_1 (g^0) = \varphi_2 (g^0) =\varphi_1 (g^N)=0$ and $\varphi_2 (g^N) =1$. Link formation costs are given by $c_{12}=0$ and $c_{21}=1$. Hence, we can derive that under two-sided link formation costs that $\varphi^a _{i}(g^0)=0$ as well as $\varphi^a_{i}(g^N)=0$, for $i=1,2$.
\\
Clearly, the empty network $g^0$ is both (strong) link deletion proof for the net payoff function $\varphi^{a}$ and supported by the superfluous Nash equilibrium characterised by $\ell_{12}=1$ and $\ell_{21}=0$. Of course, $g^0$ is also supported as a Nash equilibrium through its non-superfluous strategy profile $\ell^0_{12} = \ell^0_{21} =0$ in $(A^a, \pi^a)$.
\end{example}

\subsection{One-sided link formation costs}

It is a natural extension to consider a network formation process under a one-sided cost structure. In this approach, one of the two linking players acts as the \emph{initiator\/} and sends an initiation message to the other. If the other player, called the  \emph{responder\/}, chooses to reciprocate positively, the link materialises; otherwise, not. This link formation process has a similar nature as the process considered in \citet{BalaGoyal2000a}, except that here the responder has to consent to the formation of the link, while in Bala-Goyal's model this is not required. There the initiator can create a link with the respondent in the absence of consent.

The decision making process is more complex than that under two-sided link formation costs. Consequently, the action set has to be constructed differently. Following \citet{GillesChakrabarti2012}, for each player $i$, we introduce a strategy set given by
\begin{equation}
A_{i}^{b}= \left\{\, \left. (l_{ij},r_{ij})_{j\neq i} \, \right| \, l_{ij},r_{ij}\in \{0,1\}\, \right\}.
\end{equation}%
This means that player $i$ chooses to act as an initiator in forming a link with $j$ if she initiates a message to $j$ indicated as $l_{ij}=1$. In this case, player $j$ acts as the respondent and responds positively to this initiative if $r_{ji}=1$. On the other hand, player $j$ rejects the initiated link with $i$ if $r_{ji}=0$. Therefore, a link is only established if the initiated link is accepted, i.e., if $l_{ij}=r_{ji}=1$. This is formalised as follows.

Let $A^{b}=\prod_{i\in N}A_{i}^{b}$ be the set of such communication profiles. Given the link formation process set out above, for any profile $(l,r)\in A^{b}$, the resulting network is now given by 
\begin{equation}
g^{b}(l,r)=\{ ij\in g^N \mid l_{ij}=r_{ji}=1\}.
\end{equation}%
To delineate the one-sided model from the two-sided model, it is preferred to use a different notation for the incurred link formation costs. Instead, I introduce the function $\gamma \colon N \times N \to \mathbb R_+$ as the one-sided link formation cost structure. Here, when $i$ initiates a link with $j$---represented by $l_{ij}=1$---$i$ incurs a cost of $\gamma_{ij}\geqslant 0$, regardless of whether the initialised link is accepted by $j$ or not. On the other hand, responding to a link initialisation message is costless, i.e., $j$ incurs no cost in responding to any message $\ell_{ij}$ sent by $i$ in the link formation process. 

For a given network payoff function $\varphi$ on $N$ this now results in the following net payoff function for player $i$: 
\begin{equation}
\pi _{i}^{b}(l,r) = \varphi_{i}\left( g^{b}(l,r) \right) -\sum_{j\neq i} l_{ij} \cdot \gamma_{ij}.
\end{equation}
Formally, let $\varphi$ be a network payoff function on $N$ and let $\gamma \colon N \times N \to \mathbb R_+$ be a given one-sided link formation cost structure. Then we refer to the non-cooperative game in strategic form $\Gamma^b_\varphi ( \gamma ) = ( \mathcal A^{b},\pi ^{b})$ as the \emph{consent model of network formation with one-sided link formation costs}. 

\paragraph{Nash equilibria of the consent model with one-sided link formation costs.}

As before, we can now introduce a non-superfluous strategy tuples in the consent model with one-sided link formation costs:
\begin{definition}
Let $\varphi$ be a network payoff function on $N$ and let $\gamma \colon N \times N \to \mathbb R_+$ be a given one-sided link formation cost structure. Consider the corresponding consent model with one-sided link formation costs $\Gamma^b_\varphi ( \gamma ) = ( \mathcal A^{b},\pi ^{b})$.
\\
Then a strategy profile $(l,r) \in A^b$ is \textbf{non-superfluous} if for all pairs $i,j\in N$ it holds that 
\begin{align}
l_{ij}=1\mbox{ implies that }r_{ji}=1& \mbox{ as well as } l_{ji}=r_{ij}=0,\mbox{ and} \\
r_{ij}=1\mbox{ implies that }l_{ji}=1& \mbox{ as well as } l_{ij}=r_{ji}=0.
\end{align}
\end{definition}
Unlike for the consent model with two-sided link formation costs, each network is no longer supported by a unique non-superfluous strategy profile. Indeed, it depends on who of the two players involved initiates and who responds in the link formation process.

On the other hand, under a non-superfluous strategy profile, only one player bears the establishment cost of each existing link, and every initialisation is responded to positively. As a first step in the analysis of this one-sided approach, I explore the relationship between the Nash equilibria of the two-sided and the one-sided model. Secondly, I present a full characterisation of the Nash equilibria of the one-sided model in terms of network stability properties. These results are taken from \citet{GCS2011}.

The main question to be considered here is whether there is a network payoff function which would provide equivalence between Nash equilibria of the one-sided model and strong link deletion proofness with regard to a payoff function in a similar fashion as Theorem \ref{9:thm:2-equiv} for two-sided link formation costs. In particular, I follow efficiency logic and consider a payoff function which only assigns link formation costs to the player with the lower cost of link formation. If link formation costs are equal, a tie-breaking rule is applied.

Let $M_{i}(g)=\{j\in N_{i}(g) \mid \gamma_{ij}< \gamma_{ji}$ or $\gamma_{ij}=\gamma_{ji}, \, i<j \} \subset N_{i}(g)$ be the potential links that player $i$ should finance based on incurring the lowest link formation costs. The corresponding payoff function $\varphi ^{b}$ is defined for $i\in N$ by
\begin{equation*}
\varphi _{i}^{b}(g) = \varphi _{i}(g)-\underset{j\in M _{i}(g)}{\sum } \gamma_{ij}
\end{equation*}%
given the network payoff function $\varphi$ representing benefits without taking into account costs of link formation. We can show the following implication, which proof is relegated to Appendix A.5.
\begin{theorem} \label{9:thm:1-compare}
Let $\varphi$ be a network payoff function on $N$ and let $\gamma \colon N \times N \to \mathbb R_+$ be a given one-sided link formation cost structure. If network $g \in \mathbb{G}^N$ is strong link deletion proof for the net payoff function $\varphi^{b}$, then $g$ can be supported by a non-superfluous Nash equilibrium in the consent model with one-sided link formation costs $\Gamma^b_\varphi ( \gamma ) = (A^b, \pi^b)$.
\end{theorem}
The converse of Theorem \ref{9:thm:1-compare} does not hold as shown by the following counter-example.
\begin{example} \label{9:ex:simplo}
Consider the minimal binary network formation situation with $N= \{ 1,2 \}$ and network payoffs given by $\varphi_1 (g^0)= \varphi_2 (g^0) =0$, $\varphi_1 (g^N) =2$ and $\varphi_2 (g^N)=10$. Link formation costs are given by $\gamma_{12}=5$ and $\gamma_{21}=7$.
\\
Hence for $i=1,2$, $\varphi_{i}^{b}(g^0)=0$, $\varphi_{1}^{b}(g^N)=-3$ and $\varphi_{2}^{b}(g^N)=3$. Clearly, the complete network $g^N$ is not link deletion proof for the network payoff function $\varphi ^{b}$, since player 1 would benefit from severing the unique link $12$.
\\
However, there is a Nash equilibrium of the one-sided consent model $\Gamma^b_\varphi ( \gamma ) = (A^b, \pi^b )$ that supports the complete network $g^N \colon l_{12}=0; \, r_{12}=1; \, l_{21}=1; \, r_{21}=0$.\footnote{Note that in the case of two-sided link formation costs, the cost of link formation is a total of $\gamma_{12} + \gamma_{21} = 7+5=12$, which clearly makes the complete network $g^N$ not being supported by a Nash equilibrium in $\Gamma^a_\varphi ( \gamma )$. This indicates the underlying reason why two-sided link formation costs shrink the set of supported networks in comparison with the case of one-sided link formation costs.}
\end{example}
One might expect that a network payoff function that assigns a link initiator role to the player with the higher marginal net benefits as a result of formation of the link in question might resolve the issue of characterising the supported equilibrium networks in $\Gamma^b_\varphi ( \gamma ) = (A^b , \pi^b )$. Below it is shown that this is actually not the case.
\begin{example} \label{9:ex:two-step}
Consider a situation with three players, $N= \{ 1,2,3 \}$. The following table gives the benefits for each of the three players in the case of the formation of one of only three relevant networks:

\begin{center}
\begin{tabular}{|c|ccc|}
\hline
\textbf{Network $g$} & $\varphi _{1}(g)$ & $\varphi _{2}(g)$ & $\varphi_{3}(g)$ \\ \hline
$\{12\}$ & 10 & 10 & 0 \\ 
$\{13\}$ & 10 & 0 & 10 \\ 
$\{12,13\}$ & 15 & 20 & 20 \\ \hline
\end{tabular}
\end{center}

\noindent All other networks generate no benefits to any of the three players, i.e., $\varphi _{i}(g)=0$ for all other networks $g$ not listed in the table.
\\
Consider the following one-sided link formation cost structure: $\gamma_{12} = \gamma_{13}=9$, $\gamma_{21}=10$, $\gamma_{31}=10$, and $\gamma_{23}= \gamma_{32}= 10$. Within this context, player 1 has the highest marginal net benefit from forming links $12$ as well as $13$, namely $\varphi _{1}(\{12\})-\gamma_{12}=\varphi _{1}(\{13\})-\gamma_{13}=1$, while the other players have no positive marginal benefits from forming links $12$ and $13$.
\\
Now, the network $\{12,13\}$ is not link deletion proof for the network payoff function that is based on the property that the player with the highest net marginal benefit is assumed to finance the formation of a link. Indeed, player 1---who has the highest net marginal benefits from both links---has a negative net return from forming network $\{12,13\}$ and would prefer to sever one of the two links to increase her net benefit to $1$.
\\
On the other hand, $\{12,13\}$ is supported by a non-superfluous Nash equilibrium strategy profile under one-sided link formation costs with $l_{21}=r_{12}=1$ and $l_{31}=r_{13}=1$.
\end{example}
These examples show that the problem of finding a reasonable payoff function that completely characterises all Nash equilibria of the one-sided consent model in terms of network stability remains open. The issues are such that it can be argued that there is actually no reasonable network payoff function that characterises all supported equilibrium networks in the consent model under one-sided link formation costs.

\paragraph{Multi-stage network formation under one-sided link formation costs.}

One can ask whether certain other approaches can resolve the coordination and free riding issues that are indicated in the discussion of the converse of Theorem \ref{9:thm:1-compare} above.\footnote{This discussion requires knowledge of multi-stage, sequential games and the notion of subgame perfection. This discussion can be skipped without any difficulty. For more elaborate discussion of multi-stage and sequential games I refer to \citet{Osborne2004}, \citet{Harrington2008} and \citet{Maschler2013}.} 

Here, I consider a two-stage network formation process to restore equivalence between equilibria of that model under one-sided costs and strong link deletion proofness with respect to some well-constructed network payoff function. This is motivated by the fact that often sequential decision making solves coordination problems. With this in mind, consider the following natural two-stage process:
\begin{numm}
\item In the first stage, every players $i \in N$ initiates links by selecting initiation messages $(l_{ij})_{j\neq i}$.

\item In the second stage, all players respond to links initiated in the first stage and select $(r_{ij} \colon l_{ji} =1)_{j\neq i}$.
\end{numm}
The question is whether the subgame perfect Nash equilibria of this game are strong link deletion proof with regard to $\varphi ^{b}$. We show that this is not necessarily the case.
\begin{example}
Reconsider the simple binary linking situation in Example \ref{9:ex:simplo}. We showed earlier that the complete network $g^N = \{ 12 \}$ is not (strong) link deletion proof for the net payoff function $\varphi ^{b}$ but that there is a Nash equilibrium communication profile of the one-sided model that supports it, namely, $l_{12}=0;\, r_{12}=1;\, l_{21}=1;\, r_{21}=0$.
\\
We now show that in the two-stage network formation process described above, this communication profile is subgame perfect as well. Consider the reduced game in the second stage, given that $l_{12}=0$ and $l_{21}=1$ has been chosen in the first stage. In normal form it can now be represented as the matrix game

\begin{center}
\begin{tabular}{|c|c|c|}
\hline
\begin{tabular}{cc}
& $r_{21}$ \\ 
$r_{12}$ & 
\end{tabular}
& 0 & 1 \\
\hline
0 & $0,-7$ & $0,-7$ \\
\hline
1 & $2,3$ & $2,3$ \\
\hline
\end{tabular}
\end{center}

\noindent There are two Nash equilibria in this game, one of which is $r_{12}=1$ and $r_{21}=0$. This is exactly the second part of the indicated communication profile. Thus, the given communication profile is indeed a subgame perfect equilibrium in the two-stage link formation process.
\end{example}
The reason why sequential decision making cannot resolve the coordination problem is that here the problem stems from costs not being transferable. Complete transferability of costs and benefits would take us into the framework of \citet{JacksonWolinsky1996} and, in particular, \citet{BlochJackson2006,BlochJackson2007}.

\paragraph{A formal comparison of one-sided and two-sided link formation costs.}

Since the two models that we considered in this section have different philosophical bases, we must make some simplifying assumptions to enable a more formal comparison. In particular, we have to address how the two different link formation cost formulations are related. This simply requires us to formulate the one-sided cost structure $\gamma$ in terms of the two-sided cost structure $c$. Hence, we consider $\gamma$ to be a particular functional form of $c$.

I look at two simplified cases that facilitate this comparison.

\paragraph{\textsc{Case A:} The initiator bears all.}

Suppose that the initiator in the model with one-sided costs bears both his cost and the cost of the responder in the context of the two-sided consent model. So, initiation is tantamount to bearing the total cost of link formation, i.e., $\gamma_{ij}=c_{ij}+c_{ji}$ for all $i\neq j$. Benefits described by $\varphi$ remain individualised and are not transferable.

In this case, it is quite obvious that the Nash equilibria of the two models are not comparable, which is shown in the next simple example.
\begin{example}
Consider again a binary link formation situation with $N=\{1,2\}$ and $\varphi _{i}(g^N)=51$, $\varphi _{i}(g^0)=0$, $i=1,2 $. Moreover, let $c_{12}=c_{21}=50$. Hence, $\gamma_{12} = \gamma_{21} = 100$. Then, $g^N = \{ 12 \}$ is supported by a Nash equilibria of the two-sided model, namely through $\ell_{12} = \ell_{21}=1$. But there is no Nash equilibrium in the one-sided model that would support it because no one would be willing to pay a cost of $100$ in order to sustain this link.
\\
Next, modify the situation to let $\varphi _{1}(g^N)=12$, $\varphi _{2}(g^N)=2$, $\varphi _{i}(g^0)=0$, $i=1,2$ and $c_{12}=c_{21}=5$. Hence, $\gamma_{12} = \gamma_{21} = 10$. Then, $g^N = \{ 12 \}$ is now supported by a Nash equilibrium of the one-sided model, namely through $l_{12}=r_{21}=1$, $l_{21}=r_{12}=0$. The strategy supporting this network is not a Nash equilibrium in the two-sided model.
\end{example}

\paragraph{\textsc{Case B:} A sunk cost formulation.}

Next, we consider the case in which the link formation costs are not transferable and that the initiator has to bear only his own cost. This corresponds to a scenario where the costs of the responding party are sunk and, thus, not relevant to the decision making process.

Hence, we assume that $\gamma_{ij}=c_{ij}$ for all $i\neq j$. In this case, it can be shown that networks supported by Nash equilibria of the two two-sided model are also supported by some Nash equilibrium of the one-sided model, while the converse does not hold. For a proof of the next theorem I refer to Appendix A.6.
\begin{theorem} \label{9:thm:2-compare}
Let $\varphi$ be a network payoff function on player set $N$ and let $c \colon N \times N \to \mathbb R_+$ a two-sided link formation cost structure on $N$. 
\\
If a network $g \in \mathbb G^N$ is supported by a Nash equilibrium of the consent model with two-sided link formation costs $\Gamma^a_\varphi (c)$, then there exists a non-superfluous Nash equilibrium supporting network $g$ in the consent model with one-sided link formation costs $\Gamma^b_\varphi (c)$, i.e., for one-sided link formation cost structure $\gamma$ given by $\gamma_{ij} = c_{ij}$ for all $i,j \in N$.
\end{theorem}
We show that the converse of Theorem \ref{9:thm:2-compare} does not hold.
\begin{example}
Consider again the binary link formation situation with $N = \{ 1,2 \}$. Furthermore, assume now that $\varphi_1(g^0) = \varphi_2 (g^0) =0$, $\varphi_1 (g^N) =6$ and $\varphi_2 (g^N)=4$. Let two-sided costs of link formation be uniform, given by $c_{ij} = 5$ for all $i,j \in N$.
\\
The complete network $g^N = \{ 12 \}$ initiated by player 1 is supported by a Nash equilibrium in the one-sided model for $\gamma_{ij} = c_{ij}$. But the strategy tuple $\ell_{12} = \ell_{21}=1$ in the two-sided model that supports this network is not a Nash equilibrium in that model.
\end{example}
This discussion shows that one-sided link formation processes require a very careful analysis and do not necessarily result in very delineated conclusions.

\section{Trust and network formation}

In this section I review some concepts that try to capture the fundamental idea that ``trust builds networks''. These concepts go beyond the approaches that I have reviewed thus far, being Myerson's model and its variations as well as the Jackson-Wolinsky approach to incorporate cooperative conceptions into a network formation setting.

I discuss two different implementations of trusting behaviour into network formation. First, \citet{Buskens2005} consider the notion of \emph{unilateral stability} that is founded on the principle that players attempt the formation of links even if their correspondents did not signal that they would necessarily agree to the formation of these links. Thus, players follow the rule that one should certainly try to form links if one expects the correspondent to benefit from its formation. This leads to a refinement of the class of M-networks.

A very similar conception has been developed by \citet{GillesSarangi2010}. Within the consent model under two-sided link formation costs \citet{GillesSarangi2010} developed a belief-based stability concept denoted as \emph{monadic stability} for understanding a purely non-cooperative process of network formation based on trusting behaviour. Again players are assumed to pursue the formation of links if they perceive the correspondents to benefit from their creation. However, monadic stability is defined as a self-confirming equilibrium \citep{FudenbergLevine1993} based on these belief systems, deviating considerably from \citet{Buskens2005}'s conception of trusting behaviour.

\subsection{Unilateral stability}

The mathematical sociologists \citet{Buskens2005} proposed a refinement of the Nash equilibrium concept that considers expanding a player's ability to affect the network that is formed in a broader way than allowed through best response rationality underlying the Nash equilibrium concept. They recognised that the multitude of Nash equilibria in the Myerson model is due to a simple (mis-)coordination problem: Players are indifferent between proposing or not proposing a link if the other player actually does not propose the link herself already. This resulted in a refinement of the Nash equilibrium concept that takes account of the idea that players trust that mutually beneficial link formation will indeed be pursued by other players.
\begin{definition}
Let $\varphi$ be a network payoff function on $N$ and consider the corresponding Myerson model $\Gamma^m_\varphi = (\mathcal A^m,\pi^m)$. A network $g \in \mathbb{G}^N$ is \textbf{unilaterally stable} if there exists a strategy profile $\ell \in A^m$ in the Myerson model with $g (\ell) =g$ such that
\begin{numm}
\item for all $i \in N$ and $\ell'_i \in A^m_i \colon \pi^m_i (\ell) \geqslant \pi^m_i (\ell'_i , \ell_{-i} )$ \emph{(Nash equilibrium condition)}, and

\item for every $i \in N$ and every alternative strategy $\ell'_i \in A^m_i$, it holds that
\[
\pi^m_i (\ell^\star ) > \pi^m_i (\ell )
\]
implies that there is some $j \in N$ with $\ell'_{ij} =1$ and $\ell_{ij}=0$ for whom
\[
\pi^m_j (\ell^\star ) < \pi^m_j (\ell ) ,
\]
where $\ell^\star \in A^m$ is given by $\ell^\star_i =\ell'_i$, $\ell^\star_{jk} = \ell_{jk}$ for $j \neq i \neq k$ and $\ell^\star_{ji} = \ell'_{ij} =1$ for $j \neq i$.
\end{numm}
\end{definition}
A network is unilaterally stable if it is supported through a Nash equilibrium in the Myerson model under the additional provision that every player can modify her direct neighbourhood provided that this modification can be constructed with the consent of her chosen neighbours. So, if $i$'s proposal would make herself better off, then all newly selected neighbours would have no objections and would not receive lower payoffs as a consequence of this modification of the network.

Unilateral stability introduces a form of trusting behaviour into the Myerson approach to network formation under mutual consent. The consent of any player's neighbours is reasoned by that player is conducted in such a way that it reflects trusting behaviour by that particular player. In some sense it introduces a \emph{bounded} form of rationality of any player in her consideration of how other players respond to changes in her behaviour. As such the notion of unilateral stability can be categorised as a model of trusting behaviour in network formation under mutual consent.

An alternative definition of unilateral stability is also possible as captured in the proposition below. It reflects the idea to add trusting behaviour to the M-network concept.

\begin{proposition}
\textbf{\emph{(An alternative definition of unilaterally stable networks)}}
\\
A network $g \in \mathbb{G}^N$ is unilaterally stability if and only if $g$ is an M-network such that for every player $i \in N$ and all link sets $h^-_i \subset L_i (g)$ and $h^+_i \subset L_i (g^N \setminus g)$ it holds that \emph{either} $\varphi_i (g-h^-_i+h^+_i) \leqslant \varphi_i (g)$ \emph{or} $\varphi_i (g-h^-_i+h^+_i) > \varphi_i (g)$ implies there is some $j \in N$ such that $ij \in h^+_i$ and $\varphi_j (g-h^-_i+h^+_i) < \varphi_j (g)$.
\end{proposition}

\smallskip\noindent
Unilateral stability is the strongest individualistic or ``monadic'' network formation concept that has been proposed in the literature. Indeed, going beyond the unilateral formation of links under consent as formulated here would actually involve active participation of multiple players.

Next, we turn to discussing some simple properties of unilateral stability.
\begin{proposition} \label{9:prop:UniStab}
Let $\varphi$ be a network payoff function on $N$ and consider the corresponding Myerson model $\Gamma^m_\varphi = (\mathcal A^m,\pi^m)$. Then the following properties hold:
\begin{abet}
\item Every unilaterally stable network is strongly pairwise stable.

\item There exist strictly pairwise stable networks that are not unilaterally stable.

\item If the network payoff structure $\varphi$ is link monotone, then $g^N \in \mathbb{G}^N$ is the unique unilaterally stable network for $\varphi$.
\end{abet}
\end{proposition}
I prove all three assertions in Proposition \ref{9:prop:UniStab} in an informal fashion, rather than a rigorous mathematical way.

First, from Proposition \ref{9:prop:Mlargeness} it follows that every M-network $g$ is strong link deletion proof. Furthermore, applying the unilateral stability condition to a single link $ij \in g$ reduces to the LAP property. This immediately shows Proposition \ref{9:prop:UniStab}(a).

Next, if the network payoff structure is link monotone, then there are no objections of any player to add more links to an existing network. Hence, the complete network $g^N$ is the only M-network that satisfies the unilateral stability condition, implying the assertion stated as Proposition \ref{9:prop:UniStab}(c).

Finally, to show Proposition \ref{9:prop:UniStab}(b), I device an example for the case of three players. This example also has an important role to assess the relationship between unilateral stability and other stability concepts, introduced further down in these lecture notes.
\begin{example} \label{9:ex:unil}
Here, consider three players $N = \{ 1,2,3 \}$ and a network payoff structure $\varphi$ given in the next table.

\begin{center}
\begin{tabular}{|l|c|c|c||c|}
\hline
\textbf{Network} $g$ & $\varphi_1 (g)$ & $\varphi_2 (g)$ & $\varphi_3 (g)$ & \textbf{Stability} \\
\hline
$g^0 = \varnothing$ & 0 & 0 & 0 & Strongly PS \\
$g^1 = \{ 12 \}$ & 0 & 0 & 2 & Strictly PS \\
$g^2 = \{ 13 \}$ & 0 & 0 & 0 & \\
$g^3 = \{ 23 \}$ & 0 & 0 & 0 & \\
$g^4 = \{ 12,13 \}$ & -1 & 0 & 0 & \\
$g^5 = \{ 12,23 \}$ & 0 & -1 & 0 & \\
$g^6 = \{ 13,23 \}$ & 0 & 1 & 1 & \\
$g^7 = g^N$ & 3 & 3 & 3 & U-stable \\
\hline
\end{tabular}

\end{center}
Here, $g^0$ is strongly pairwise stable, but is not unilaterally stable. Indeed, player 3 can add both links 13 and 23 to make $g^6$ without objection of the other players.
\\
Furthermore, $g^1$ is strictly pairwise stable and again not unilaterally stable. As before, player 3 can add links 13 and 23 to move to $g^N$ without any objections of the other two players. This shows assertion \ref{9:prop:UniStab}(b).
\\
Also, it is clear from the table that the complete network $g^N$ is unilaterally stable, since it is strong link deletion proof. Note that in this case $g^N$ is strictly pairwise stable as well.
\\
Finally, I refer to Example \ref{9:ex:unil-mon} for a detailed discussion of an example in which assertion of Proposition \ref{9:prop:UniStab}(b) is strengthened in the sense that the class of strictly pairwise stable networks is completely disjoint from the class of unilaterally stable networks.
\end{example}
To assess unilateral stability, it is clear that \citet{Buskens2005} introduce it as an expression of firmly methodological individualistic behavioural principles:  Players act selfishly only, but conjecture that other players will consent to the creation of links that directly benefit them. It builds on the hypothesis that players offer no objections to the formation of links that directly benefit them.

However, an alternative interpretation can easily be applied here as well. Indeed, the unilateral stability concept can be interpreted to be an application of a principle of trusting behaviour: players trust others to consent to forming links if it does not hurt them. This is closely akin to the model of trusting behaviour. An alternative model of trusting behaviour founded on belief systems in Myerson's framework is discussed next.

\subsection{Monadic stability}

\citet{GillesSarangi2010} introduced a belief-based conception of trusting behaviour in the setting of the consent model with two-sided link formation costs. Their approach imposes minimal informational requirements. Unlike other models of strategic network formation, players need not be aware of the payoffs associated with every network. For any given network $g \in \mathbb{G}^N$ to emerge in such a setting, a player is required to know the payoffs associated with any change (creation or deletion) only involving their own direct links $ij \in L_i (g)$. 

This results in an amendment of Myerson's consent game such that, based on their information, players form simple, myopic beliefs about the direct benefits other players will receive from establishing links with them. According to these myopic beliefs, each player $i \in N$ assumes that another player $j \in N$ is willing to form a new link with $i$ if $j$ stands to benefit from it in the prevailing network. Similarly $i$ also assumes that $j$ will break an existing link $ij$ in the prevailing network if $j$ does not benefit from having this link. Thus, in this process player $i$ assumes that all other links in the prevailing network remain unchanged.

Therefore, these monadic beliefs are indeed ``myopic'' in the sense that they only pertain to direct effects of the addition or removal of a link in the network. Hence, these beliefs disregard higher order effects on the payoffs of all players in the network due to the addition or removal of such a link. As such these behavioural standards reflect a \emph{bounded} form of rationality in decision making, implying that the boundedly rational foundation of monadic stability is fundamentally different from the rational standard imposed by unilateral stability.

Such myopic beliefs essentially capture the idea that network formation primarily occurs between acquaintances with sufficiently large an amount of information about each other to assess first order effects of network changes.\footnote{That social relations are mainly formed between acquaintances is confirmed empirically by \citet{WellmanEtAl1988} using data from the East York area. This principle also forms the foundation of the model in \citet{Brueckner2006}, who models friendship as building links between players chosen from a given set of acquaintances.} This concept is a normal form implementation of the self-confirming equilibrium concept introduced by \citet{FudenbergLevine1993} within the setting of the Myerson model and its variations.

One can assess these myopic belief systems as reflecting a certain form of ``confidence'' on the part of each player to engage in communication to form links with other players that have an obvious (first-order) benefit from the addition of such a link. This confidence suffices to form non-trivial social networks. As stated, a certain commonality is assumed among the players in order to formulate such common priors and beliefs on which this confidence is founded. In this regard we assume that players are acquaintances and build relationships through beliefs about actions undertaken by other players.\footnote{It is clear that this approach is akin to the notion of unilateral stability introduced before. A comparison of monadic stability with unilateral stability is, therefore, called for. This is further developed here as well.}

We now formalise these myopic belief systems for the consent model under two-sided link formation costs.

\paragraph{Defining monadic stability.}

Throughout we assume there is a given network payoff function $\varphi \colon \mathbb{G}^N \to \mathbb{R}^N$ and we impose a two-sided link formation cost structure $c = ( c_{ij} )_{i,j \in N}$. Based on this data, consider the corresponding consent model under two-sided link formation costs $\Gamma^a_\varphi (c) = ( \mathcal A^a, \pi^a )$. We can introduce specific belief systems in this setting that represent the trusting behavioural principle as discussed above.
\begin{definition}
Let $\ell \in A^a$ be an arbitrary communication profile resulting in network $g = g( \ell )$. For every player $i \in N$ we define $i$'s \textbf{monadic belief system} concerning $\ell$ as a communication profile $\ell^{i \star} \in A^a$ given by
\begin{numm}
\item for every $j \neq i$ with $ij \in g$ let
\begin{itemize}
\item $\ell^{i \star}_{ji} =0$ if $\varphi_j \left( g-ij \right) +c_{ji} > \varphi_j (g)$ and
\item $\ell^{i \star}_{ji} =1$ if $\varphi_j \left( g-ij \right) +c_{ji} \leqslant \varphi_j (g)$;
\end{itemize}

\item for every $j \neq i$ with $ij \notin g$ let
\begin{itemize}
\item $\ell^{i \star}_{ji} =0$ if $\varphi_j \left( g+ij \right) -c_{ji} < \varphi_j (g)$ and
\item $\ell^{i \star}_{ji} =1$ if $\varphi_j \left( g+ij \right) -c_{ji} \geqslant \varphi_j (g)$;
\end{itemize}

\item and for all $j,k \in N$ with $j \neq i \neq k$ let $\ell^{i \star}_{jk} = \ell_{jk}$.
\end{numm}
\end{definition}
A monadic belief system reflects that a player believes that other players are myopically selfish and will act in their myopic self-interest. Hence, links are consented to if that directly benefits the other player and are refused if deleting that link benefits the other player.

Now monadic stability simply requires that each player acts rationally in view of these beliefs.
\begin{definition}
Let $\varphi$ and $c$ be given with the corresponding consent model under two-sided link formation costs $\Gamma^a_\varphi (c) = ( \mathcal A^a, \pi^a )$.
\begin{abet}
\item A network $g \in \mathbb{G}^N$ is \textbf{weakly monadically stable} for $(\varphi ,c)$ if there exists some communication profile $\ell \in A^a$ with $g = g(\ell )$ such that for every $i \in N \colon \ell_i \in A^a_i$ is a best response to her monadic beliefs $\ell^{i \star}_{-i} \in A^a_{-i}$ for payoff function $\pi^a$; thus,
\begin{equation}
	\pi^a_i \left( \, g ( \ell'_i , \ell^{i \star}_{-i}) \, \right) \leqslant \pi^a_i \left( \, g ( \ell_i , \ell^{i \star}_{-i} ) \, \right)
\end{equation}
for all $\ell'_i \in A^a_i$.

\item A network $g \in \mathbb{G}^N$ is \textbf{monadically stable} for $(\varphi ,c)$ if there exists some communication profile $\ell \in A^a$ with $g = g(\ell )$ such that for every $i \in N \colon \ell_i \in A^a_i$ is a best response to her monadic beliefs $\ell^{i \star}_{-i} \in A^a_{-i}$ for payoff function $\pi^a$ and player $i$'s monadic belief system $\ell^{i \star}$ is confirmed in the sense that for every $j \neq i$ it holds that $\ell^{i \star}_{ji} = \ell_{ji}$. 
\end{abet}
\end{definition}
Weak monadic stability of a network is founded on the principle that every player $i \in N$ anticipates---as captured by her (monadic) expectations about direct links---that other players will respond myopically selfishly to her attempts to form a link with them. Note that $\ell_{-i}$ is fully replaced by the player's belief system $\ell^{i \star}_{-i}$ in the standard best-response formulation of Nash equilibrium for player $i$ and is therefore irrelevant for the decision making process of $i$.

Monadic stability strengthens the above concept by requiring that the beliefs of each player are confirmed in the resulting equilibrium. Hence, monadic stability imposes a self-confirming condition on the weakly monadic equilibrium. This describes the situation that all players are fully satisfied with their beliefs; the observations that they make about the resulting network confirm their beliefs about the other players' payoffs. This amounts to updating one's initial beliefs. As such, monadic stability is an implementation of a \emph{self-confirming equilibrium} based on the monadic belief system in the context of consent model with two-sided link formation costs \citep{FudenbergLevine1993}.

To delineate the two monadic stability concepts for networks, we discuss a three player example. This example shows that the class of monadically stable networks is usually strictly larger than the class of the weakly monadically stable networks.
\begin{example} \label{9:ex:mon-wmon}
Consider $N = \{ 1,2,3 \}$ and assume uniform link formation costs with $c_{ij} = 1$ for all $i,j \in N$. Let the network payoff function $\varphi$ be given in the table below:

\begin{center}
\begin{tabular}{|l|c|c|c||c|}
\hline
\textbf{Network} $g$ & $\varphi_1 (g)$ & $\varphi_2 (g)$ & $\varphi_3 (g)$ & \textbf{Stability} \\
\hline
$g^0 = \varnothing$ & 0 & 0 & 0 & $M_w$ \\
$g^1 = \{ 12 \}$ & 0 & 1 & 0 & \\
$g^2 = \{ 13 \}$ & 0 & 0 & 3 & \\
$g^3 = \{ 23 \}$ & 0 & 0 & 0 & \\
$g^4 = \{ 12,13 \}$ & 3 & 0 & 0 & \\
$g^5 = \{ 12,23 \}$ & 1 & 3 & 3 & \\
$g^6 = \{ 13,23 \}$ & 2 & 2 & 5 & $M_w$ \\
$g^7 = g^N$ & 3 & 5 & 6 & $M_w$ and $M$ \\
\hline
\end{tabular}
\end{center}

\smallskip\noindent
This table identifies whether the network in question is weak monadically stable---indicated by $M_w$---or whether it is monadically stable---indicated by $M$.
\\
Within this example we now consider some of the networks given and analyse their stability properties.
\begin{description}
\item[Network $g^0$:]
We show that this network is weakly monadically stable for a supporting communication profile that is superfluous. Indeed, select $\ell_0 = ( \, (1,1) , (0,0) , (0,0) \, ) \in A^a$ with $g (\ell_0) = g^0 = \varnothing$. Observe here that player 1 incurs link formation costs with $\pi^a_1 (\ell_0) =-2$, while $\pi^a_2 (\ell_0) = \pi^a_3 (\ell_0) =0$. Then we can determine the monadic belief systems for all players as
\begin{align*}
\ell^{1 \star}_0 & = ( - , (1,0) , (1,0) \, ) \\
\ell^{2 \star}_0 & = ( (0,1), - , (0,0) \, ) \\
\ell^{3 \star}_0 & = ( (1,0) , (0,0) ,- \, )
\end{align*}
It should be emphasised that in this case player 1 believes that both other players are willing to make links with her, because there are direct benefits from forming such links. However, the other players believe that player 1 will not attempt to make a link with them, because she has no direct (net) benefits from doing so. This refers to a classical coordination problem.
\\
Now we determine that the best responses for all players are given by
\begin{itemize}
\item $\beta_1 \left( \ell^{1 \star}_0 \right) = (1,1)$ is the unique best response to $\ell^{1 \star}_0$ for player 1.

\item $\beta_2 \left( \ell^{2 \star}_0 \right) = (0,0)$ is the unique best response to $\ell^{2 \star}_0$ for player 2.

\item $\beta_3 \left( \ell^{3 \star}_0 \right) = (0,0)$ is the unique best response to $\ell^{3 \star}_0$ for player 3.
\end{itemize}
This confirms that $g^0$ is indeed weakly monadically stable for $\ell_0$. However, $g^0$ is not monadically stable, since in the communication profile $\ell_0$, player 1's beliefs are not confirmed. She expects the other two players to be willing to form links with her, although they do not do so.

\item[Network $g^5$:]
This network is neither weakly monadically stable, nor monadically stable. The non-superfluous communication profile $\ell_5 = (\, (1,0) , (1,1) , (0,1) \, )$ is an obvious candidate to support this network. For this profile we compute that
\begin{align*}
\ell^{1 \star}_5 & = ( - , (1,1) , (1,1) \, ) \\
\ell^{2 \star}_5 & = ( (1,0), - , (0,1) \, ) \\
\ell^{3 \star}_5 & = ( (1,1) , (1,1) ,- \, )
\end{align*}
This results into the following best response configuration:
\begin{itemize}
\item $\beta_1 \left( \ell^{1 \star}_5 \right) = (1,1)$ is the unique best response to $\ell^{1 \star}_5$ for player 1.

\item $\beta_2 \left( \ell^{2 \star}_5 \right) = (1,1)$ is the unique best response to $\ell^{2 \star}_5$ for player 2.

\item $\beta_3 \left( \ell^{3 \star}_5 \right) = (1,1)$ is the unique best response to $\ell^{3 \star}_5$ for player 3.
\end{itemize}
From this it is clear that $g^5$ cannot be supported by $\ell_5$. This illustrates that weak monadic stability requires selecting a best response to a \emph{specific} set of beliefs for each player $i \in N$. Without such a restriction on the beliefs it would be possible to support any strategy as weakly monadic stable. Moreover, observe that players only form beliefs about the behaviour of their acquaintances with regard to direct links, making it myopic but realistic. In fact, because of this, it is possible that monadically stable equilibria do not exist.
\\
Finally, we can complete the argument by checking that other communication profiles can be ruled out in similar fashion.

\item[Network $g^6$:]
We argue that this network is weakly monadically stable as well. We can show that $g^6$ is supported by the action tuple $\ell_6 = \left( \, (0, 1), (1, 1, ), (1, 1) \, \right)$. Again we compute
\begin{align*}
\ell^{1 \star}_6 & = ( - , (1,1) , (1,1) \, ) \\
\ell^{2 \star}_6 & = ( (1,1), - , (1,1) \, ) \\
\ell^{3 \star}_6 & = ( (0,1) , (1,1) ,- \, )
\end{align*}
Note here that player 1 is indifferent between $g^6$ and $g^7$ in terms of her net payoff $\pi^a$. Thus, in the computation of $\ell^{2 \star}_6$ we use the bias of player 1 towards having more links rather than fewer in player 2's belief system.
\\
This results into the following best response configuration:
\begin{itemize}
\item $\beta_1 \left( \ell^{1 \star}_6 \right) = \{ \, (0,1),(1,1) \, \}$ is the set of best responses to $\ell^{1 \star}_6$ for player 1, i.e., $(0,1)$ and $(1,1)$ are both best responses for this player.

\item $\beta_2 \left( \ell^{2 \star}_6 \right) = (1,1)$ is the unique best response to $\ell^{2 \star}_6$ for player 2.

\item $\beta_3 \left( \ell^{3 \star}_6 \right) = (1,1)$ is the unique best response to $\ell^{3 \star}_6$ for player 3.
\end{itemize}
This shows that $\ell_6$ is indeed supported as a weak monadically stable communication profile. On the other hand, $g^6$ is not monadically stable, since the beliefs of player 2 are not confirmed.

\item[Network $g^7$:]
First, we claim that this network is strictly pairwise stable. Strong link deletion proofness follows trivially from the payoffs listed. Indeed, the net payoffs in other networks ($g^0, \ldots ,g^6$) are at most the net payoff in $g^7$ for all players. Second, strict link addition proofness is trivially satisfied since there are no links that are not part of $g^7 = g^N$.
\\
Furthermore, the complete network $g^7 = g^N$ is weakly monadically stable. We claim that $g^7$ is supported by the only communication profile supporting this network, $\ell_7 = ( \, (1, 1), (1, 1, ), (1, 1) \, )$. We can determine that the monadic belief systems are given by
\begin{align*}
\ell^{1 \star}_7 & = ( - , (1,1) , (1,1) \, ) \\
\ell^{2 \star}_7 & = ( (1,1), - , (1,1) \, ) \\
\ell^{3 \star}_7 & = ( (1,1) , (1,1) ,- \, )
\end{align*}
From this we conclude that
\begin{itemize}
\item $\beta_1 \left( \ell^{1 \star}_7 \right) = \{ \, (0,1),(1,1) \, \}$ is the set of best responses to $\ell^{1 \star}_7$ for player 1.

\item $\beta_2 \left( \ell^{2 \star}_7 \right) = (1,1)$ is the unique best response to $\ell^{2 \star}_7$ for player 2.

\item $\beta_3 \left( \ell^{3 \star}_7 \right) = (1,1)$ is the unique best response to $\ell^{3 \star}_7$ for player 3.
\end{itemize}
So, $\ell_7$ is indeed a best response profile with regard to the generated monadic belief systems. Hence, $g^7$ is indeed weakly monadically stable.
\\
Finally, all players' monadic belief systems are confirmed here. So, in fact, $g^7$ is monadically stable.
\end{description}
In this example, it is made clear that the introduced monadic belief systems require only that players use minimal information about each other's payoffs to formulate appropriate expectations about each other's linking behaviour. Indeed, monadic stability only considers players to use first-order effects of forming new links and deleting existing links to formulate their monadic beliefs.
\end{example}
This example clarifies the relationship between the notion of weak monadic stability and the monadic stability concept. Next, I provide a more general characterisation.
\begin{proposition} \label{9:prop:mon-wmon}
Let the network payoff function $\varphi$ and the link formation cost structure $c$ be given. Every monadically stable network $g \in \mathbb{G}^N$ for $(\varphi ,c)$ satisfies the following two properties:
\begin{numm}
	\item $g$ is weakly monadically stable, and
	\item $g$ is supported by a monadic belief system $\ell^g$ that is non-superfluous in the sense that $\ell^g_{ij} = \ell^g_{ji}$ for all pairs $i,j \in N$.
\end{numm}
\end{proposition}
\begin{proof}
Let $g \in \mathbb{G}^N$ be monadically stable and let action tuple $\ell^g \in A^a$ support $g$ as such. Suppose that $ij \notin g$ with $\ell^g_{ij} = 1$ and $\ell^g_{ji} = 0$. Then from the property that $\ell^g_i \in A^a_i$ is a best response to the belief system $\ell^{g \, i \star}_{-i}$ it can be concluded that $\ell^g_{ij} = 1$ implies that $\ell^{g \, i \star}_{ji} = 1$. But this would then imply that $\ell^g_{ji} \neq \ell^{g \, i \star}_{ji}$, violating the monadic stability self-confirmation condition.
\end{proof}

\medskip\noindent
The reverse of the assertion of Proposition \ref{9:prop:mon-wmon} is not true. Simple examples can be constructed in which weakly monadically stable networks exist that satisfy the stated property, but which are not monadically stable.

A few comments regarding the relationship between weak monadic stability and network-based stability concepts are in order here. First, weakly monadically stable networks are not necessarily strong link deletion proof or link addition proof. Second, a network that is strong link deletion proof as well as link addition proof is not necessarily weakly monadically stable. We refer to network $g^6$ in Example \ref{9:ex:mon-wmon}, which is weakly monadically stable, but not link addition proof. The other comparisons can also be shown by properly constructed counterexamples.

\paragraph{An equivalence result.}

The main insight from this approach is that trust indeed builds very strong networks. This is exemplified by the equivalence of the class of monadically stable and strictly pairwise stable networks. For a proof I refer to Appendix A.7.
\begin{theorem} \label{9:equiv:mon-sps}
Let the network payoff function $\varphi$ and the link formation cost structure $c = (c_{ij} )_{i,j \in N}$ be given such that $c_{ij} >0$ for all $i,j \in N$ with $i \neq j$. Then a network $g \in \mathbb G^N$ is monadically stable for $(\varphi ,c)$ if and only if $g$ is strictly pairwise stable for the network payoff function $\varphi^a$ given by
\begin{equation}
	\varphi^a_i (g) = \varphi_i (g) - \sum_{ij \in L_i (g)} c_{ij}
\end{equation}
\end{theorem}
Through the monadic stability concept we have considered the notion of confidence---as a form of mutual trust---into an advanced equilibrium concept, specifically designed for network formation. Confidence is introduced as an \emph{internalised} feature into the behaviour of the players in network formation. Thus, trusting behaviour is as such a individualised feature rather than a social normative phenomenon. 

The strength as well as the weakness of the monadic stability approach is the myopic nature of the belief systems. Players do not apply very sophisticated reasoning; they only look at the first order effects of link formation. Natural future extensions of this line of theoretical research should explore the possibility of introducing forward looking behaviour to understand how farsightedly stable networks arise.\footnote{This can be compared with existing models of farsighted network formation developed in \citet{Deroan:2003p1371}, \citet{Dutta:2005p1105}, \citet{PageWooders2002}, \citet{Herings:2009p1144}, \citet{Navarro2014}, \citet{LimitedFarsighted2016}, \citet{Forster2016} and \citet{Schaar2020}.}

\subsection{A comparison of unilateral and monadic stability}

As mentioned in the introduction to this section, unilateral and monadic stability seem to be founded on the same principles of trusting behaviour: Players attempt to form links with other players if they perceive these players to benefit from these links. 

Recall that a network is unilaterally stable if there is no player who can induce changes to the network based on the belief that other players will consent to these changes if they are not harmful to them. Note here that unilateral stability assumes a fully rational form of farsightedness in the decision making process: All proposed changes to the network---as made by a single player---are fully taken into account by all involved players before consent is granted. Thus, unilateral stability assumes a sophisticated level of rational forecasting by all players, who need to consent to the proposed changes to the network.

This implies that unilateral stability is indeed founded on the principle of trusting behaviour. Implicitly, players are indeed acting on beliefs that other players will act in their self-interest when confronted with proposed changes to their link sets. As such, unilateral stability is a trust equilibrium concept.

On the other hand, monadic stability assumes a much less sophisticated form of rational decision making. Indeed, players are actually assumed to be boundedly rational: Players form monadic beliefs that only take first-order changes to the payoffs of other individualised into account. So, if a player proposes to add multiple links, her beliefs are founded on payoff changes per addition of a single link rather than the complete set of links. Beliefs are, thus, founded on a bounded form of reasoning by these players.

Moreover, only after beliefs are formed, all players base their actions on maximising their payoffs given these boundedly rational monadic beliefs. There can arise a build-in mismatch of beliefs and actual outcomes in the form of realised changes to the network. However, actual actions need to confirm the monadic beliefs of players. This pushes the decision making process from unrealistic to justified, since these beliefs are observed by the player decision makers.

Therefore, monadic stability is a trust equilibrium concept as well and is designed explicitly to be based on an embedded form of trusting behaviour in the disguise of belief formation on trusting principles. These trusting principles are not violated due to the confirmation condition in the monadic stability concept---in contrast to the weak monadic stability notion.

In summary, monadic stability is founded on a boundedly rational form of trusting behaviour. This contrasts with unilateral stability in which all decisions are based on a more farsighted, rational implementation of similar ideas.

\paragraph{A formal comparison.}

Next I consider a more technical comparison of the two concepts. From the discussion above it cannot be expected that the application of monadic stability and unilateral stability results in exactly the same class of stable network. The next example shows that these two conceptions can lead to completely different sets of stable networks.
\begin{example} \label{9:ex:unil-mon}
Again consider the by-now familiar case of three players $N = \{ 1,2,3 \}$. Let the network payoff function $\varphi$ be given in the table below and assume that link formation is costless, i.e., $c_{ij} =0$ for all $i,j \in N$.

\begin{center}
\begin{tabular}{|l|c|c|c||c|}
\hline
\textbf{Network} $g$ & $\varphi_1 (g)$ & $\varphi_2 (g)$ & $\varphi_3 (g)$ & \textbf{Stability} \\
\hline
$g^0 = \varnothing$ & 0 & 0 & 0 & \\
$g^1 = \{ 12 \}$ & 1 & 1 & 2 & M-stable \\
$g^2 = \{ 13 \}$ & 0 & 0 & 0 & \\
$g^3 = \{ 23 \}$ & 0 & 0 & 0 & \\
$g^4 = \{ 12,13 \}$ & 0 & 0 & 1 & \\
$g^5 = \{ 12,23 \}$ & 0 & 0 & 1 & \\
$g^6 = \{ 13,23 \}$ & 3 & 3 & 3 & U-stable \\
$g^7 = g^N$ & 4 & 2 & 4 & \\
\hline
\end{tabular}
\end{center}

\noindent
The table reports the stability properties of the various networks. There emerge three interesting networks to be investigated, namely $g^1$, $g^6$ and $g^7 =g^N$. I discuss these in detail below:
\begin{description}
\item[Network $g^1$:] We investigate the stability properties of this network. First, note that $g^1$ is not unilaterally stable. Indeed, player 3 prefers to propose the formation of links $13$ and $23$ to create network $g^N$, which represents a strict Pareto improvement for all players in $N$.
\\
Second, network $g^1$ is supported by a non-superfluous communication profile that is represented as $\ell^1 = ( \, (1,0), (1,0) , (0,0) \, )$. This results into a monadic belief system given by
\begin{align*}
\ell^{1 \star}_1 & = ( - , (1,0) , (0,0) \, ) \\
\ell^{2 \star}_1 & = ( (1,0), - , (0,0) \, ) \\
\ell^{3 \star}_1 & = ( (1,0) , (1,0) ,- \, )
\end{align*}
Clearly $\ell^1$ constitutes a best response profile to the given monadic belief system and the monadic belief system is confirmed through $\ell^1$, showing that $g^1$ is supported as a monadically stable network.\footnote{Similarly, note that $g^1$ is actually a strictly pairwise stable network. The equivalence theorem shows that, therefore, $g^1$ has to be monadically stable.}

\item[Network $g^6$:] First, note that $g^6$ is strongly pairwise stable as well as unilaterally stable. Indeed, only player 1 has an incentive to add link $12$ to form the complete network $g^7 = g^N$, which is rejected by player 2 due to a loss in payoff. There are no players who have incentives to sever any of the two existing links.
\\
Next, $g^6$ is not monadically stable. Indeed, take the non-superfluous communication profile that supports it, given by $\ell^6 = ( \, (0,1), (0,1), (1,1) \, )$. Then the corresponding monadic belief system is
\begin{align*}
\ell^{1 \star}_6 & = ( - , (0,1) , (1,1) \, ) \\
\ell^{2 \star}_6 & = ( (1,1), - , (1,1) \, ) \\
\ell^{3 \star}_6 & = ( (0,1) , (0,1) ,- \, )
\end{align*}
Obviously, the communication profile $\ell^6$ is a best response to the monadic belief system above. This implies that $g^6$ is weakly monadically stable. However, it is \emph{not} monadically stable. Indeed, player 2 believes that player 1 would pursue the creation of a link with her---as represented by $\ell^{2 \star}_{12} =1$. This is not as described by $\ell^6$; player 1 does not propose a link to player 2 and, as such, the belief system of player 2 is not confirmed in the equilibrium communication profile.

\item[Network $g^7 = g^N$:] To conclude the discussion of the situation described in this example, we consider the complete network $g^7 = g^N$, which is uniquely supported by the communication profile $\ell^7 = ( \, (1,1),(1,1),(1,1) \, )$. The resulting monadic belief systems can now be represented by
\begin{align*}
\ell^{1 \star}_7 & = ( - , (0,1) , (1,1) \, ) \\
\ell^{2 \star}_7 & = ( (1,1), - , (1,1) \, ) \\
\ell^{3 \star}_7 & = ( (1,1) , (1,1) ,- \, )
\end{align*}
Obviously, the communication strategy $\ell^7_1 = (1,1)$ is not a best response to $\ell^{1 \star}_7$, since player 1 expects player 2 not to form a link with her. Therefore, $\ell^7$ is not supported as a monadically stable communication profile. Thus, $g^7$ is not weakly monadically stable.
\\
Furthermore, this network is neither unilaterally stable; in particular, it is not link deletion proof. Indeed, player 2 has an incentive to break the link with player 1 to move to network $g^6$.
\end{description}
This example clearly shows that the class of unilaterally stable networks can be completely disjoint from the class of monadically stable networks. In this example, however, the unilaterally stable network is weakly monadically stable. This implies that in a unilaterally stable network monadic beliefs can destabilise the network, leading to unending improvement attempts by the players in the network. Thus, boundedly rational belief formation can undermine a farsightedly rational foundation for the network; as such, it represents an example of a direct conflict between farsighted or full and boundedly rational behaviour.
\end{example}

\subsection{Existence of monadically stable networks}

The question of existence of monadically stable networks is an important one. The previous discussion already identified the class of monadically stable networks to be exactly equal to the class of strictly pairwise stable networks. Obviously, this class is empty for a large collection of network payoff structures. Here I investigate certain conditions under which the class of monadic networks is non-empty. 

These conditions are related to the notion of a network potential as seminally developed by \citet{GillesChakrabarti2007}. There it is explored what the consequences are of founding network payoffs on an underlying link-based payoff function---denoted as a \emph{network potential}. Network payoff functions that admit a potential impose a payoff structure in which players assess the value of links in a similar fashion. It can be shown that for network payoff structures that are founded on such potentials, there exist strictly pairwise stable networks.

In the subsequent discussion, I summarise the main insights from \citet{GillesChakrabarti2007}. For details of the proofs of the main theorems I also refer to that paper and its appendices. Before stating the main definitions and the resulting properties, I recall the definition of two potential concepts in the context of a non-cooperative game $( \mathcal A, \pi )$ on the player set $N$ as seminally introduced by \citet{MondererShapley1996}.
\begin{definition}
	Let $( \mathcal A, \pi )$ be a non-cooperative game on player set $N$. Then: 
	\begin{abet}
	\item The game $( \mathcal A, \pi )$ \textbf{admits an exact potential} in the sense of \citet{MondererShapley1996} if there exists a function $P \colon A \to \mathbb R$ such that
	\begin{equation}
		\pi_i (a) - \pi_i (b_i , a_{-i}) = P (a) - P (b_i , a_{-i})
	\end{equation}
	for every player $i \in N$, every strategy tuple $a \in A$ and every strategy $b_i \in A_i$.
	
	\item The game $( \mathcal A, \pi )$ \textbf{admits an ordinal potential} in the sense of \citet{MondererShapley1996} if there exists a function $P \colon A \to \mathbb R$ such that
	\begin{equation}
		\pi_i (a) > \pi_i (b_i , a_{-i}) \quad \mbox{if and only if} \quad P (a) > P (b_i , a_{-i})
	\end{equation}
	for every player $i \in N$, every strategy tuple $a \in A$ and every strategy $b_i \in A_i$.
	\end{abet}
\end{definition}
Based on these two notions of game-theoretic potentials, we can now consider how network payoff structures might be founded on similar constructs. 

\paragraph{Network potentials.}

There are two main conceptions of the notion of a potential as a founding device in the determination of network payoffs. Again we refer to these notions as an ``exact potential'' and an ``ordinal potential'', following the accepted terminology in the literature. The next definition introduces these two notions.
\begin{definition}
	Let $\varphi \colon \mathbb G^N \to \mathbb R^N$ be a network payoff function.
	\begin{abet}
		\item The network payoff function $\varphi$ \textbf{admits an exact potential} if there exists a function $\Lambda \colon \mathbb G^N \to \mathbb R$ such that
		\begin{equation}
			\varphi_i (g) - \varphi_i (g-ij) = \Lambda (g) - \Lambda (g-ij)
		\end{equation}
		for every network $g \in \mathbb G^N$, every player $i \in N$ and every link $ij \in L_i (g)$.
		\item The network payoff function $\varphi$ \textbf{admits an ordinal potential} if there exists a function $\Lambda \colon \mathbb G^N \to \mathbb R$ such that the following conditions hold:
		\begin{align}
			\varphi_i (g) > \varphi_i (g-ij) \quad \mbox{if and only if} \quad \Lambda (g) > \Lambda (g-ij) \\
			\varphi_i (g) < \varphi_i (g-ij) \quad \mbox{if and only if} \quad \Lambda (g) < \Lambda (g-ij) \\
			\varphi_i (g) = \varphi_i (g-ij) \quad \mbox{if and only if} \quad \Lambda (g) = \Lambda (g-ij)
		\end{align}
		for every network $g \in \mathbb G^N$, every player $i \in N$ and every link $ij \in L_i (g)$.
	\end{abet}
\end{definition}
An exact potential imposes that the network payoff structure exhibits a \emph{cardinally uniform} way of how players assess the addition or deletion of a link to a network. It is clear that the admittance of an exact potential is a very strong condition on the network payoff structure. This is confirmed by the following insight from \citet[Theorem 3.3]{GillesChakrabarti2007}:
\begin{lemma}
	A network payoff function $\varphi$ admits an exact potential if and only if the corresponding Myerson model $\Gamma^m_\varphi$ admits an exact potential in the sense of \citet{MondererShapley1996}.
\end{lemma}
The admittance of an ordinal potential in a network payoff structure imposes a uniform assessment of deleting and adding links to networks by all players in purely \emph{ordinal} terms. Although this property is significantly weaker than the admittance of an exact potential, it remain a rather demanding condition on the network payoff structure. The next lemma makes clear that there is again a relationship with the notion of an ordinal potential in the sense of \citet{MondererShapley1996}. The next lemma is stated as Theorem 4.3 in \citet{GillesChakrabarti2007}. For a proof I refer to that source.
\begin{lemma} \label{lemma:OrdinalP}
	Let $\varphi$ be some network payoff structure. If the corresponding Myerson model $\Gamma^m_\varphi$ admits an ordinal   potential in the sense of \citet{MondererShapley1996}, then $\varphi$ admits an ordinal potential.
\end{lemma}
The reverse of the assertion stated in Lemma \ref{lemma:OrdinalP} is not true, as shown in \citet[Example 4.4]{GillesChakrabarti2007}.

\paragraph{Properties of network payoff structures that admit potentials.}

Using the introduced notions of game-theoretic and network potentials, we can now distinguish three essential classes of network payoff structures. First, those network payoff structures that admit an exact potential; second, those network payoff structures for which the corresponding Myerson game admits an ordinal potential; and, finally, those network payoff structures that admit an ordinal potential. Each of these classes is larger than the previous.

The next propositions collect some properties of the third class, namely those network payoff structures that admit an ordinal potential. For proofs of these assertions I again refer to \citet{GillesChakrabarti2007}.
\begin{proposition}
	Let $\varphi$ be some network payoff structure that admits an ordinal potential $\Lambda$. Then the following properties hold:
	\begin{numm}
		\item There exists at least one pairwise stable network.
		\item The sets of strongly pairwise stable and strictly pairwise stable networks coincide.
	\end{numm}
\end{proposition}
The class of network payoff structures for which the corresponding Myerson game admits an ordinal potential is particularly interesting. Indeed, \citet[Theorem 5.7]{GillesChakrabarti2007} show that for this class of network payoff structures there exist strictly pairwise stable networks. I state for completeness the complete assertion:
\begin{proposition}
	Let $\varphi$ be a network payoff function for which the corresponding Myerson model $\Gamma^m_\varphi$ admits an ordinal   potential in the sense of \citet{MondererShapley1996}. Then there exists at least one strictly pairwise stable network for $\varphi$.
\end{proposition}
This property gives rise to the main conclusion regarding the existence of a monadically stable network. Indeed, the admittance of an ordinal potential in the Myerson model gives rise to the existence of a strictly pairwise stable network, which in turn is monadically stable due to the fundamental equivalence theorem. As a consequence, we can formulate the following main existence theorem:
\begin{theorem}
	Let $\varphi \colon \mathbb G^N \to \mathbb R^N$ be a network payoff structure and let $c \colon N \times N \to \mathbb R_+$ be a link formation cost structure. If the corresponding consent model with two-sided link formation costs $\Gamma^a_\varphi (c)$ admits an ordinal potential in the sense of \citet{MondererShapley1996}, then there exists at least one monadically stable network for $(\varphi ,c)$.
\end{theorem}

\section{Correlated network formation}

The previous section focussed mainly on the internalisation of trust in the behaviour of players to result into so-called ``trusting behaviour'' in link formation. We chose to internalise trusting behaviour in the form of belief systems (monadic stability) or through stability concepts themselves (unilateral stability). However, there is rather different an approach possible in which trusting behaviour is explicitly modelled through an externally determined institutional arrangement. These institutional arrangements are implemented collectively and are endowed with a form of collectively accepted self-enforcement.

In my discussion I mainly considered behavioural rules that can be viewed as being part of a trusted governance system. All players are assumed to be embedded in such a governance system, expressing this in the formulated monadic stability concept as embedded monadic belief systems. Hence, we use game theoretic concepts to give this embeddedness an explicit, institutional form as a generally accepted behavioural rule, to behave according to the stated monadic belief system.

\paragraph{Correlation devices.}

Next, I turn to a much more explicit conception of behavioural sociality. One can model guiding behavioural norms also as being \emph{external} to the players, rather than fully internalised---as is the case for the notion of monadic belief systems. This refers to the possibility to let external ``devices'' guide and coordinate decision-making in a game theoretic setting. In particular, one can consider the question: ``Can external guidance let decision makers achieve a higher payoff than that is achieved through the set of supported Nash equilibria?''

A seminal study by \citet{Aumann1974} introduced an innovative way to exactly introduce a formal way to establish mutually beneficial coordination among players. These external arrangements are denoted as \emph{correlation devices}. The basic idea is that the decisions made by players are influenced by things that are external to the decision problem itself, but are situated in their immediate surrounding. The classical example is that of a traffic light.\footnote{The following discussion is mainly based on the excellent account of correlated equilibrium in Chapter 9 of \citet{Maschler2013}. I recommend the interested reader to look at their presentation.}

The game theoretic representation is a form of the \emph{Game of Chicken} as explored extensively in the literature. Two drivers approach a road crossing. At the crossing, each driver can either ``stop'' (action $S$) or ``continue'' (action $C$). If both continue there will result a crash; if both stop, both look foolish and need to coordinate their passing through prolonged negotiation (with hand gestures); and if one stops and the other continues, there is regret of the stopper and maximal payoff to the one who continues. The resulting payoffs can be captured by the following game-theoretic payoff matrix:

\begin{center}
\begin{tabular}{l|c|c|}
 & \textbf{S} & \textbf{C} \\
 \hline
\textbf{S} & 5,5 & 2,7 \\
\textbf{C} & 7,2 & 0,0 \\
\hline
\end{tabular}
\end{center}

\noindent
There result three Nash equilibria in mixed strategies here, namely one driver stops and the other continues---resulting in payoffs $(7,2)$ and $(2,7)$ depending on who actually stops---and the case in which both players stop or continue with equal probability---resulting into the expected payoff vector $\left( 3 \tfrac{1}{2}, 3 \tfrac{1}{2} \right)$. The latter includes a probability of $\tfrac{1}{4}$ of a crash, due to both players continuing.

Now consider that there is an outside regulator---represented as a correlation device---added to this situation in the form of a traffic light. The most important assumption of this arrangement is that both drivers are fully informed about what fraction of time the traffic light is in what colour. Hence, both drivers know the probability distribution that is implemented through the traffic light. We investigate two traffic light arrangements:
\begin{itemize}
\item First, consider that with equal probability the traffic light gives a red light to one player and a green light to the other. Adopting the normal rule to stop for red and to continue for green, we actually coordinate between the two Nash equilibria $(S,C)$ and $(C,S)$, resulting into an expected payoff computed as
\[
\mathbb{E} \, \pi_1 = \tfrac{1}{2} (7,2) + \tfrac{1}{2} (2,7) = \left( 4 \tfrac{1}{2}, 4 \tfrac{1}{2} \right) .
\]
Here there no positive probability of a crash and both drivers are reasonably content with their expected payoff.
\\
Would this traffic light be \emph{self-enforceable} within the given social decision situation? We need to check whether this traffic light arrangement is indeed beneficial to both player drivers if it is implemented as suggested by these two drivers. Obviously, if any of these drivers deviates from the recommendation, while the other follows it, there is a crash---resulting in zero payoffs. So, the suggested arrangement is indeed self-enforcing.

\item In comparison with our regular traffic light, we can even increase the expected payoff by introducing a more complicated coordination device. Indeed, consider a traffic light that can stop both drivers simultaneously with a given probability. In that case, the drivers negotiate themselves and proceed with caution. So, the traffic light can give both drivers simultaneously the signal ``red'', at which both drivers are suggested to stop and proceed with caution.
\\
This allows the mixing of three outcomes in this decision situation. Suppose now that the traffic light gives both drivers simultaneously ``red'' with probability $\tfrac{1}{2}$ and one driver ``red'' and the other driver ``green'' with equal probabilities $\tfrac{1}{4}$.\footnote{This means that both drivers get \emph{private} recommendations from the traffic light; they do not know what the colour to the other driver is. This is the usual arrangement in modern traffic law.} We can depict the resulting probability distribution over all outcomes in a probability matrix:

\begin{center}
\begin{tabular}{l|c|c|}
 & \textbf{S} & \textbf{C} \\
 \hline
\textbf{S} & $\tfrac{1}{2}$ & $\tfrac{1}{4}$ \\
\textbf{C} & $\tfrac{1}{4}$ & 0 \\
\hline
\end{tabular}
\end{center}

\noindent
The resulting expected payoffs can now be computed as
\[
\mathbb{E} \, \pi_2 = \tfrac{1}{2} (5,5) + \tfrac{1}{4} (7,2) + \tfrac{1}{4} (2,7) = \left( \, 4 \tfrac{3}{4} , 4 \tfrac{3}{4} \, \right) \gg \left( \, 4 \tfrac{1}{2}, 4 \tfrac{1}{2} \, \right) = \mathbb{E} \, \pi_1.
\]
Again we can ask whether this traffic light is self-enforcing. If one driver receives ``red'', he knows that the other driver receives ``red'' with probability $\tfrac{2}{3}$ and ``green'' with probability $\tfrac{1}{3}$. So, if he continues there is a crash with probability $\tfrac{1}{3}$ and he receives an expected payoff of $\tfrac{1}{3} \cdot 0 + \tfrac{2}{3} \cdot 7 = 4 \tfrac{2}{3} < 4 \tfrac{3}{4}$, the latter being the expected payoff if he follows the recommendation of the traffic light. Again, we conclude that the traffic light arrangement is indeed self-enforcing; no player has an incentive to deviate from the provided arrangement and recommendations.
\end{itemize}
One can ask whether this reasoning can be extended to even higher payoffs. Indeed, Aumann showed that this is the case up to payoff level 5. The arrangement that both drivers always face a red light---that is, ``red'' with probability 1---is, of course, not self-enforcing.

\paragraph{Using correlation devices in network formation.}

Correlation devices can also be introduced in the processes of network formation. I return to the network formation process under consent that we discussed thus far and consider how external correlation devices in the form of external recommender systems can guide players to form ``good'' networks. We first take a look at a by-now familiar network formation situation with three players.
\begin{example}
As before, let $N = \{ 1,2,3 \}$ be the set of three players. Also, we choose $\varphi$ to be a minor modification of the network payoff function studied in Example \ref{9:ex:unil-mon}, given in the table below, and again we assume that link formation is costless, i.e., $c_{ij} =0$ for all $i,j \in N$.
\\
As reported in the table below, there are actually five M-networks, namely all strong link deletion proof networks given by $\mathcal{M} = \{ g^0, g^1, g^2, g^3, g^6 \}$. These five M-networks correspond only to three payoff vectors, namely $(0,0,0)$, $(1,1,2)$ and $(3,3,3)$.

\begin{center}
\begin{tabular}{|l|c|c|c||c|}
\hline
\textbf{Network} $g$ & $\varphi_1 (g)$ & $\varphi_2 (g)$ & $\varphi_3 (g)$ & \textbf{M-network} \\
\hline
$g^0 = \varnothing$ & 0 & 0 & 0 & M \\
$g^1 = \{ 12 \}$ & 1 & 1 & 2 & M \\
$g^2 = \{ 13 \}$ & 0 & 0 & 0 & M \\
$g^3 = \{ 23 \}$ & 0 & 0 & 0 & M \\
$g^4 = \{ 12,13 \}$ & 8 & 8 & 1 & \\
$g^5 = \{ 12,23 \}$ & 0 & 0 & 1 & \\
$g^6 = \{ 13,23 \}$ & 3 & 3 & 3 & M \\
$g^7 = g^N$ & 4 & 2 & 4 & \\
\hline
\end{tabular}
\end{center}

\noindent
The main question I consider here is: Can we introduce a correlation device in this network formation situation that results in higher expected payoffs than those from the high-paying M-network $g^6$? Indeed, $g^6$ is the most obvious M-network that the players can aim for. Therefore, the payoff vector $(3,3,3)$ acts as a benchmark in relationship to any correlation device.
\\
Consider an external recommender system based on the three networks $g^4$, $g^6$, and $g^7=g^N$. In particular, assume that this correlation device recommends (i) all three players to execute signalling strategy $\ell^a = \left( \, (1,1), (1,0), (1,0) \, \right)$ resulting in network $g^4$ with probability $\alpha = \tfrac{1}{12}$; (ii) the signalling strategy $\ell^b = \left( \, (0,1), (0,1), (1,1) \, \right)$ resulting in network $g^6$ with probability $\beta = \tfrac{2}{3}$; and (iii) the signalling strategy $\ell^c = \left( \, (1,1), (1,1), (1,1) \, \right)$ resulting in network $g^7= g^N$ with probability $\gamma = \tfrac{1}{4}$. The expected payoffs under this system are now given by
\begin{align*}
	\mathbb{E} \, \pi  \left( \ell \right) & = \alpha \cdot \varphi (g^4) + \beta \cdot \varphi (g^6) + \gamma \cdot \varphi (g^7) \\
	& = \tfrac{1}{12} \cdot \varphi (g^4) + \tfrac{2}{3} \cdot \varphi (g^6) + \tfrac{1}{4} \cdot \varphi (g^7) = \left( \, 3 \tfrac{2}{3} , 3 \tfrac{1}{6} , 3 \tfrac{1}{12} \right) \gg (3,3,3) = \varphi (g^6)
\end{align*}
Hence, coordinating the link building actions through this recommender system results into a \emph{strict} Pareto improvement over the best M-network. It remains to show that all three players have no incentives to deviate from the recommended correlated strategy:
\begin{itemize}
	\item Player 1: The only plausible alternative signalling strategy is to play $\ell_1 = (1,1)$ to achieve the high paying network $g^4$. This results actually in no changes to the recommended networks, due to the recommended strategies executed by the two other players under the selected correlation device. Hence, player 1 has no gain from deviating from the recommended strategy.
	\item Player 2: The only plausible alternative signalling strategy for this player is to execute $\ell'_2 = (1,0)$ to establish network $g^4$. But this results in a lower expected payoff for player 2 if the other players follow the recommended strategies in $\ell \colon$
		\[
	           \mathbb E \, \varphi_2 \left( \ell'_2 \right) = \tfrac{1}{12} \cdot \varphi_2 (g^4) + \tfrac{2}{3} \cdot \varphi_2 (g^2) + \tfrac{1}{4} \cdot \varphi_2 (g^4) = \tfrac{1}{3} \cdot 8 + \tfrac{2}{3} \cdot 0 = 2 \tfrac{2}{3} < 3 \tfrac{1}{6} = \mathbb E \, \pi_2 \left( \ell \right) .
	   \]
\end{itemize}
This recommender system uses two non-M-networks, $g^4$ and $g^7 = g^N$. Therefore, this correlation device is founded on considerations outside the realm of the stability concepts that we have considered thus far. It shows that inefficient networks and non-stable networks play a role in network formation processes.
\end{example}
The example above shows just a single application of the correlated equilibrium concept to network formation analysis. The application of this concept opens the way to further exploration, even though the multitude of correlated equilibria is discouraging. Indeed, Aumann showed that the collection of expected payoff vectors supported by correlated equilibria includes the convex hull of all Nash equilibrium payoff vectors. This is rather daunting and discouraging from the perspective that correlation will not lead to a smaller class of supported networks.

However, the main research question that is still open is whether there exists a specific class of correlation devices that could guide players to highly productive networks. Throughout our history, humans have in fact found ways to implement very effective correlation devices to build effective and high-paying networks. This includes recommender systems such as job recommendation referrals and socio-economic recommendations through friendship networks. Further exploration of these systems from a Aumannian perspective is required to develop a theory that interprets these practical systems as correlation devices.

%\newpage
\singlespace
\bibliographystyle{econometrica}
\bibliography{RPDB}

\newpage
\appendix

\section{Proofs of the main theorems}

\subsection{Proof of Theorem \ref{9:equiv:Deletion}.}

\textbf{If:} Let $\varphi$ be  convex on $\mathcal{D}(\varphi )$. Obviously from the definitions and the discussions it follows that $\mathcal{D}_s(\varphi ) \subset \mathcal{D}(\varphi )$. Thus, we only have to show that $\mathcal{D}(\varphi ) \subset \mathcal{D}_s(\varphi )$.
\\
Now let $g \in \mathcal{D}(\varphi )$. Then for every player $i \in N$ and link $ij \in L_i(g)$ it has to hold that $\varphi_i(g) \geqslant \varphi_i (g - ij)$ due to link deletion proofness of $g$. In particular, for any link set $h \subset L_i(g) \colon \sum_{ij \in h} [ \varphi_i (g)-\varphi_i (g-ij) \, ] \geqslant 0$. Since $\varphi$ is convex on $\mathcal{D} (\varphi )$ and $g \in \mathcal{D}(\varphi )$, it follows that $\varphi_i(g) \geqslant \varphi_i(g - h)$ for every link set $h \subset L_i(g)$. In other words, $g$ is strong link deletion proof, i.e., $g \in \mathcal{D}_s(\varphi )$.
\\[1em]
\textbf{Only if:} Assume that $\mathcal{D}(\varphi ) = \mathcal{D}_s(\varphi )$. Suppose further to the contrary that the payoff structure $\varphi$ is not convex on $\mathcal{D}(\varphi )$. Then there exists some network $g \in \mathcal{D}(\varphi )$ and some player $i \in N$ such that for some link set $h \subset L_i(g)$ we have that $\sum_{ij \in h} [ \varphi_i (g)- \varphi_i (g-ij) \, ] \geqslant 0$ as well as $\varphi_i (g) < \varphi_i (g-h)$. But then this implies straightforwardly that player $i$ would prefer to sever all links in $h$, i.e., $g \notin \mathcal{D}_s(\varphi )$. Thus, $g$ cannot be strong link deletion proof giving us the necessary contradiction.
\\[1em]
This completes the proof of the assertion of Theorem \ref{9:equiv:Deletion}.

\subsection{Proof of Theorem \ref{9:equiv:Addition}.}

Assertion (a) is trivial and a proof is therefore omitted.
\\[1em]
\textsc{Proof of (b).}
\\
\textbf{If:} Let $\varphi$ be discerning on $\mathcal{A} (\varphi )$. Suppose that $g$ is LAP. Furthermore, assume that $i,j \in N$ with $ij \notin g$ are such that $\varphi_i (g+ij) \geqslant \varphi_i (g)$. Now, if $\varphi_j (g+ij) = \varphi_j (g)$, then by definition of $\varphi$ being discerning, $\varphi_i (g+ij) > \varphi_i (g)$. This contradicts the hypothesis that $g$ is LAP. Thus, $\varphi_j (g+ij) < \varphi_j (g)$, confirming that $g$ is indeed $\star$-LAP.
\\[0.5em]
\textbf{Only if:} Suppose that $\varphi$ is not discerning on $\mathcal{A} (\varphi )$. Then there exists some network $g$ that is LAP and for some $i,j \in N$ with $ij \notin g$ it holds that $\varphi_i (g+ij) = \varphi_i (g)$ as well as $\varphi_j (g+ij) = \varphi_j (g)$. But this immediately implies that $g$ can in fact not be $\star$-LAP, since the link $ij$ should be in $g$ for it to be $\star$-LAP. This is a contradiction.
\\[1em]
\textsc{Proof of (c).}
\\
\textbf{If:} Suppose that $\varphi$ is uniform on $\mathcal{A}_{\star} (\varphi )$ and take some $g \in \mathcal{A}_{\star} (\varphi )$. Assume that $i,j \in N$ with $ij \notin g$. Then first suppose that
\begin{equation}
\label{eqn:10-GS}
\varphi_i (g) \leqslant \varphi_i (g+ij) .
\end{equation}
Then by $g$ being $\star$-LAP it has to hold that
\begin{equation}
\label{eqn:11-GS}
\varphi_j (g) > \varphi_j (g+ij) .
\end{equation}
But also by uniformity of $\varphi$ it has to hold that
\begin{equation}
\label{eqn:12-GS}
\varphi_j (g) \leqslant \varphi_j (g+ij) .
\end{equation}
But (\ref{eqn:11-GS}) is in direct contradiction to (\ref{eqn:12-GS}). Thus, we conclude that (\ref{eqn:10-GS}) cannot hold. Therefore, for any $ij \notin g$ it has to hold that $\varphi_i (g) > \varphi_i (g+ij)$ as well as $\varphi_j (g) > \varphi_j (g+ij)$. Hence, we conclude that $g$ is actually SLAP, i.e., $g \in \mathcal{A}_s (\varphi )$.
\\[0.5em]
\textbf{Only if:} Assume that $\mathcal{A}_s (\varphi ) = \mathcal{A}_{\star} (\varphi )$. Now take $g \in \mathcal{A}_{\star} (\varphi )$ to be $\star$-LAP. Then from $g$ being SLAP, it follows that $\varphi_i (g) > \varphi_i (g+ij)$ as well as $\varphi_j (g) > \varphi_j (g+ij)$. This implies that $\varphi$ indeed has to be uniform for $g$.
\\[1em]
This proves the assertion of Theorem \ref{9:equiv:Addition}.

\subsection{Proof of Theorem \ref{9:prop:Mlargeness}}

First, we show assertion (a). \\ Suppose that there is an M-network $g \in \mathbb{G}^N$ supported by a Nash equilibrium strategy profile $\ell \in A^m$ that is not strong link deletion proof. Then there is some $i \in N$ and $h_i \subset L_i (g)$ with $\varphi_i (g-h_i) > \varphi_i (g)$. But then player $i$ can modify his linking strategy as $\ell'_{ij} = 0$ if $ij \in h_i$ and $\ell'_{ij} = \ell_{ij}$. Then $g ( \ell'_i , \ell_{-i} ) = g-h_i$ implying that $\pi^m_i ( \ell'_i , \ell_{-i} ) > \pi^m (\ell)$. Therefore, $\ell$ cannot be a Nash equilibrium in $(\mathcal A^m,\pi^m)$. This is a contradiction, showing that M-networks are strong link deletion proof.
\\
Next, let $g \in \mathcal D_s ( \varphi )$ be a strong link deletion proof network for the network payoff function $\varphi$ on $N$. Suppose that $g$ is not an M-network. Then the corresponding signalling tuple $\ell^g$---where $\ell^g_{ij} =1$ if $ij \in g$ and $\ell^g_{ij}=0$ otherwise---is not a Nash equilibrium strategy tuple in the Myerson model $\Gamma^m_\varphi$. Hence, there is a player $i \in N$ and an alternative strategy $\ell_i \in A_i$ with $\ell_i \neq \ell^g_i$ such that $\pi^m_i (\ell^g) < \pi^m_i ( \ell_i , \ell^g_{-i})$. If we denote by $h+i = \{ ij \mid \ell^g_{ij}=1$ and $\ell_{ij}=0 \}$, then it is clear that $g ( \ell_i , \ell^g_{-i}) = g-h_i \subset L_i (g)$. Using the definition of the Myerson payoff function $\pi^m$, we have established that $\varphi_i (g) < \varphi_i (g-h_i)$, which contradicts the hypothesis that $g$ is strong link deletion proof.
\\[1ex]
To show assertion (b), suppose that $\varphi$ is link monotone. Take any network $g \in \mathbb{G}^N$ and construct a strategy profile $\ell^g \in A^m$ by $\ell^g_{ij} =1$ if and only if $ij \in g$, for all $i,j \in N$. It is easy to see that $\ell^g$ is indeed a Nash equilibrium in $(\mathcal A^m,\pi^m)$ due to $\varphi$ being link monotone: For any $i \in N$, any deviation $\ell_i$ from $\ell^g_i$ induces the link set $L_i \left( \,  g ( \ell_i , \ell^g_{-i} ) \, \right) \subseteq L_i (g)$ for $i$. This implies by link monotonicity that $\pi^m_i ( \ell_i , \ell^g_{-i} ) = \varphi_i \left( \, g ( \ell_i , \ell^g_{-i} ) \, \right) \leqslant \varphi_i (g) = \pi^m (\ell^g)$.

\subsection{Proof of Theorem \ref{9:thm:2-equiv}.}

\textbf{(a) implies (c):} Let $\ell^{\star }$ be an arbitrary Nash equilibrium in $( \mathcal A^{a},\pi ^{a})$. Then denote $g^{\star } =  g^{m}(\ell^{\star })=\{ij\in g^N\mid \ell_{ij}^{\star } = \ell_{ji}^{\star}=1\}$. We show that $g^{\star }$ is strong link deletion proof for the derived network payoff function $\varphi^a$.
\\
Suppose player $i$ deletes a certain link set $h_{i}\subset L_{i}(g^{\star })$. Define $\ell_{i}\in A_{i}^{a}$ as $\ell_{ij}=1$ if $ij\in g^{\star }-h_{i}$ and $\ell_{ij}=0$ for $ij\notin g^{\star }-h_{i}$. Then by $\ell^{\star}$ being a Nash equilibrium in $(A^a, \pi^a)$ it follows that $g^{m}(\ell_{i},\ell_{-i}^{\star }) = g^{\star }-h_{i}$ and $\pi _{i}^{a}( \ell^{\star}) \geqslant \pi _{i}^{a}( \ell_{i}, \ell_{-i}^{\star })$. Hence, 
\begin{align*}
\varphi _{i}^{a}(g^{\star })& =\varphi _{i}(g^{\star })-\underset{j\in N_{i}(g^{\star })}{\sum }c_{ij}=\pi _{i}^{a}( \ell^{\star })+\underset{k\colon \ell_{ik}^{\star }=1, \ell_{ki}^{\star }=0}{\sum }c_{ik} \\
& \geqslant \pi _{i}^{a}( \ell^{\star })\geqslant \pi_{i}^{a}( \ell_{i}, \ell_{-i}^{\star })=\varphi _{i}(g^{m}( \ell_{i}, \ell_{-i}^{\star }))-\underset{k\neq i}{\sum } \ell_{ik} \cdot c_{ik} \\
& =\varphi _{i}(g^{\star }-h_{i})-\underset{k\in N_{i}(g^{\star }-h_{i})}{\sum }c_{ik}=\varphi _{i}^{a}(g^{\star }-h_{i}).
\end{align*}%
This proves that $g^{\star }$ is strong link deletion proof for $\varphi ^{a}$.

\medskip \noindent \textbf{(c) implies (b):} Suppose that $g^{\star}\subset g^N$ is a strong link deletion proof network for $\varphi ^{a}$. We show that it is supported by a non-superfluous Nash equilibrium strategy in $( \mathcal A^{a},\pi ^{a})$. Consider the unique non-superfluous strategy profile $\ell^{\star }\in A^{a}$ such that $g^{m} \left( \ell^{\star } \right)= g^{\star }$. We proceed to show that $\ell^{\star }$ is a Nash equilibrium in $( \mathcal A^{a},\pi ^{a})$ and $\ell^{\star}_{ij}=1$ if and only if $ij \in g^{\star}$. Indeed, 
\begin{equation*}
\pi _{i}^{a}( \ell^{\star })=\varphi _{i}(g^{m}( \ell^{\star }))-\underset{k\neq i}{\sum } \ell_{ik}^{\star }\cdot c_{ik} = \varphi _{i}(g^{\star })-\sum_{k \in N_{i}(g^{\star })}c_{ik}=\varphi _{i}^{a}(g^{\star })\text{.}
\end{equation*}%
Next, for some player $i$ consider some deviation $\ell_{i}\neq \ell_{i}^{\star }$. Define $h_{i}= \{ ik \in g^{\star }\mid \ell_{ik}=0\}$. Then, $g^{m}( \ell_{i}, \ell_{-i}^{\star})=g^{\star }-h_{i}.$ Since $g^{\star }$ is strong link deletion proof with respect to $\varphi ^{a}$, it follows that $\varphi _{i}^{a}(g^{\star}-h_{i})\leqslant \varphi _{i}^{a}(g^{\star })$. Thus, 
\begin{align*}
\pi _{i}^{a}(\ell_{i}, \ell_{-i}^{\star })& =\varphi _{i}(g^{m}( \ell_{i}, \ell_{-i}^{\star}))-\underset{k\neq i}{\sum } \ell_{ik}\cdot c_{ik} \\
& =\varphi _{i}(g^{\star }-h_{i})-\underset{k\in N_{i}(g^{\star }-h_{i})}{\sum }c_{ik}-\underset{k\colon \ell_{ik}=1, \ell_{ki}^{\star }=0}{\sum }c_{ik} \\
& \leqslant \varphi _{i}(g^{\star }-h_{i})-\underset{k\in N_{i}(g^{\star}-h_{i})}{\sum }c_{ik} \\
& =\varphi _{i}^{a}(g^{\star }-h_{i})\leqslant \varphi _{i}^{a}(g^{\star})=\pi _{i}^{a}(l^{\star })\text{.}
\end{align*}%
This proves that the non-superfluous signal profile $\ell^{\star }$ is indeed a Nash equilibrium.

\medskip\noindent Trivially (b) implies (a), which proves the assertion and completes the proof of Theorem \ref{9:thm:2-equiv}.

\subsection{Proof of Theorem \ref{9:thm:1-compare}}

Let $g^{\star }$ be strong link deletion proof under the net payoff function $\varphi ^{b}$. For $g^{\star }$, define a non-superfluous communication profile $\lambda ^{\star }=(l^{\star },r^{\star }) \in A^b$ as follows:
\begin{numm}
\item $l_{ij}^{\star }=r_{ji}^{\star }=1$ if $ij\in g^{\star}$ and $\gamma_{ij} < \gamma_{ji}$, or

\item $l_{ij}^{\star }=r_{ji}^{\star }=1$ if $ij\in g^{\star}$, $\gamma_{ij}=\gamma_{ji}$ and $i<j$, or 

\item $l_{ij}^{\star }=r_{ji}^{\star }=0$ if $ij\notin g^{\star}$.
\end{numm}
Obviously, $g^{b}\left( l^{\star },r^{\star } \right) =g^{\star }$ and 
\begin{equation*}
\pi _{i}^{b}(\lambda ^{\star })=\varphi _{i} \left( g^{b}(\lambda ^{\star }) \right) - \underset{j\neq i}{\sum }l_{ij}^{\star }\cdot \gamma_{ij}=\varphi _{i}(g^{\star })-\underset{j\in M _{i}(g^{\star })}{\sum } \gamma_{ij}=\varphi _{i}^{b}(g^{\star }).
\end{equation*}%
Now, for player $i \in N$ consider an arbitrary deviation $\widehat{\lambda }_{i}= \left( \, \widehat{l}_{i},\widehat{r}_{i} \right) \neq \left( l_{i}^{\star},r_{i}^{\star } \right) = \lambda _{i}^{\star }$. In any such deviation, no new links will be formed because if $ij \notin g^{\star },$ it follows that  $l_{ji}^{\star }=r_{ji}^{\star }=0$. However, links in $i$'s neighbourhood link set $L_i (g^{\star} )$ can be deleted. Hence, let $g^{b} \left( \, \widehat{\lambda }_{i},\lambda _{-i}^{\star } \right) =g^{\star }-h_{i}$ where $h_{i}\subset L_{i}(g^{\star })$.
\\
We prove that $j\in N_{i}(g^{\star }-h_{i})$ and $\left[ \,\gamma_{ij}<\gamma_{ji} \mbox{ or } \gamma_{ij}=\gamma_{ji}, \ i<j\,\right]$ implies that $\widehat{l}_{ij}=1$. In other words, $j\in M _{i}(g^{\star }-h_{i})\subset N_{i}(g^{\star }-h_{i})$ implies that $\widehat{l}_{ij}=1$.
\\
Now, assume by contradiction that for some $j\in M _{i}(g^{\star }-h_{i}) \colon \widehat{l}_{ij}=0$. Now, 
\begin{equation}
j\in N_{i}(g^{\star }-h_{i})\Leftrightarrow \widehat{l}_{ij}=1\text{
and }r_{ji}^{\star }=1\text{ or }\widehat{r}_{ij}=1\text{ and }l%
_{ji}^{\star }=1.  \label{eqnO}
\end{equation}%
But $l_{ji}^{\star }=1$ implies by construction that $\gamma_{ij} > \gamma_{ji}$ or $\gamma_{ij}=\gamma_{ji}, \ i>j$. Furthermore, $r_{ji}^{\star }=1$ implies by construction that $\gamma_{ij}<\gamma_{ji}$ or $\gamma_{ij}=\gamma_{ji}$, $i<j$. Since $\widehat{l}_{ij}=0,$ by (\ref{eqnO}), it follows that $\widehat{r}_{ij}=l_{ji}^{\star }=1$ which implies that $\gamma_{ij}>\gamma_{ji}$ or $\gamma_{ij}=\gamma_{ji}$ with $i>j$. This contradicts $j\in M _{i}(g^{\star }-h_{i})$ completing the proof of the claim stated above.
\\
Now, the proven claim implies that
\begin{equation} \label{eqnR}
\underset{j\in M _{i}(g^{\star }-h_{i})}{\sum }\ \gamma_{ij} \leqslant \underset{j\in N_{i}(g^{\star }-h_{i})}{\sum } \widehat{l}_{ij}\cdot \gamma_{ij} \leqslant \sum_{j\neq i} \widehat{l}_{ij} \cdot \gamma_{ij}.
\end{equation}
Hence, 
\begin{align*}
\pi _{i}^{b} \left( \, \widehat{\lambda }_{i},\lambda _{-i}^{\star } \, \right) & =\varphi_{i} \left( g^{b}(\widehat{\lambda }_{i},\lambda _{-i}^{\star }) \right) - \sum_{j \neq i} \widehat{l}_{ij} \cdot \gamma_{ij} = \varphi _{i} \left( g^{\star}-h_{i} \right) - \sum_{j \neq i} \widehat{l}_{ij} \cdot \gamma_{ij} \\
& \leqslant \varphi_{i} \left( g^{\star }-h_{i} \right) -\underset{j\in M_{i}(g^{\star }-h_{i})}{\sum } \gamma_{ij} = \varphi _{i}^{b}(g^{\star}-h_{i}) \\
& \leqslant \varphi _{i}^{b}(g^{\star })=\pi _{i}^{b} \left( l^{\star},r^{\star } \right) .
\end{align*}%
The first inequality follows from (\ref{eqnR}) and the second follows from the fact that $g^{\star }$ is strong link deletion proof with respect to $\varphi ^{b}$. This completes the proof of Theorem \ref{9:thm:1-compare}.

\subsection{Proof of Theorem \ref{9:thm:2-compare}}

Let $g^{\star }$ be supported by a Nash equilibrium signalling profile $\ell^{\star }\in A^{a}$ in the consent model with two-sided link formation costs $( \mathcal A^{a},\pi ^{a})$. We now construct a non-superfluous strategy tuple $\left( \, \widehat{l}, \widehat{r} \, \right) \in A^{b}$ in the consent model with one-sided link formation costs such that $g^{b} \left( \, \widehat{l}, \widehat{r} \, \right) = g^{\star }$ and $\left( \, \widehat{l}, \widehat{r} \, \right)$ is a Nash equilibrium in $( \mathcal A^{b},\pi ^{b})$.
\\
From Theorem \ref{9:thm:2-equiv}, we can assume without loss of generality
that $\ell^{\star } \in A^{a}$ is non-superfluous. Given $\ell^{\star }$, we define 
$\widehat{\lambda }= \left( \, \widehat{l},\widehat{r} \, \right) \in A^{b}$ by
\begin{numm}
\item $\widehat{l}_{ij} = \widehat{r}_{ji}=1$ and $\widehat{l}_{ji} = \widehat{r}_{ij}=0$ if and only if $\ell_{ij}^{\star } = \ell_{ji}^{\star }=1$, and either $c_{ij} < c_{ji}$, or $c_{ij} = c_{ji}$ with $i<j.$

\item $\widehat{l}_{ij} = \widehat{l}_{ji} = \widehat{r}_{ij} = \widehat{r}_{ji} =0$ if and only if $\ell_{ij}^{\star } = \ell_{ji}^{\star } =0$.
\end{numm}
It follows immediately that $\widehat{\lambda }= \left( \, \widehat{l},\widehat{r} \, \right)$ is a non-superfluous communication profile in $A^b$ supporting $g^{b} \left( \, \widehat{l}, \widehat{r} \, \right) =g^{\star }$.
\\
It remains to be shown that $\widehat{\lambda }$ is a Nash equilibrium of the consent model with one sided link formation costs. We sketch the proof of this assertion.
\\
Now, if $\widehat{\lambda }$ is not a Nash equilibrium, then it has to be because some player prefers to delete one or more of her links. Also, any link delivers the same benefit to the player as under two-sided link formation costs, while it would cost no more to establish the link. Thus, preferring to keep a link under two-sided link formation costs, implies that the player would prefer to keep the link under one-sided link formation costs. Mathematical details of this argument are left to the reader.
\\
This completes the proof of Theorem \ref{9:thm:2-compare}.

\subsection{Proof of Theorem \ref{9:equiv:mon-sps}}

We first develop some simple auxiliary insights for weakly monadically stable networks. Suppose that $g \in \mathbb{G}^N$ is weakly monadically stable relative to the data $\varphi$ and $c = ( c_{ij} )_{i,j \in N}$. Then there exists some action tuple $\hat{\ell} \in A^a$ such that $g = g (\hat{\ell})$ and for every player $i \in N \colon \hat{\ell}_i \in A^a_i$ is a best response to the monadic belief system $\hat{\ell}^{i \star}_{-i} \in A^a_{-i}$ for the payoff function $\pi^a$.

For this setting we state two auxiliary results.
\begin{lemma} \label{9:lem:61}
If $\hat{\ell}^{i \star}_{ji}=0$ and $c_{ij} > 0$, then $\ell_{ij} = 0$ is the unique best response to $\hat{\ell}^{i \star}_{-i}$.
\end{lemma}
\begin{proof}
Clearly, if player $i$ selects $\ell_{ij} = 1$, $i$ only incurs strictly positive costs $c_{i j} > 0$ and no benefits. This implies that player $i$ makes a loss from trying to establish link $ij$. Hence, $\ell_{ij} = 0$ is the unique best response to $\hat{\ell}^{i \star}_{-i}$.
\end{proof}
\begin{lemma} \label{9:lem:62}
If $ij \in g(\hat{\ell})$ with $c_{ij} >0$ as well as $c_{ji} >0$, then $\hat{\ell}^{i \star}_{ji} = \hat{\ell}^{j \star}_{ij} =1$.
\end{lemma}
\begin{proof}
We remark that $ij \in g = g(\hat{\ell})$ if and only if $\hat{\ell}_{ij} = \hat{\ell}_{ji} = 1$. The negation of the assertion stated in Lemma \ref{9:lem:61} applied to $\hat{\ell}_ij = 1$ and $\hat{\ell}_{ji} = 1$ independently now implies that $\hat{\ell}^{i \star}_{ji} = \hat{\ell}^{j \star}_{ij} =1$.
\end{proof}

\medskip\noindent
We also require a partial characterisation of weakly monadically stable networks. This is stated in the following lemma.
\begin{lemma} \label{9:lem:63}
Let the cost structure $c \gg 0$ be strictly positive. Then every weakly monadically stable network $g \in \mathbb{G}^N$ in the consent model with two-sided link formation costs $(A^a, \pi^a )$ is link deletion proof for the network payoff function $\varphi^a$.
\end{lemma}
\begin{proof}
Suppose that $g \in \mathbb{G}^N$ is weakly monadic in the consent model with two-sided link formation costs $(A^a, \pi^a )$. Then there exists some communication profile $\hat{\ell} \in A^a$ such that $g = g(\hat{\ell})$ and for every player $i \in N \colon \hat{\ell}_i \in A^a_i$ is a best response to $\hat{\ell}^{i \star}_{-i}$ for the game theoretic payoff function $\pi^a$.
\\
Suppose now that $g$ is not link deletion proof for $\varphi^a$. Then there exists some $i \in N$ with $ij \in g$ for some $j \neq i$ and $\varphi^a (g - ij) > \varphi^a_i (g)$, implying that $\varphi_i (g-ij) +c_{ij} > \varphi_i(g)$. By definition, $\hat{\ell}^{j \star}_{ij} =0$. Hence, from Lemma \ref{9:lem:61}, $\ell_{ji} =0$ is the unique best response to $\hat{\ell}^{j \star}$ for player $j$. Since $ij \in g$ by assumption it has to hold that $\hat{\ell}_{ji} = 1$. This contradicts the hypothesis that $\hat{\ell}_j$ is a best response to $\hat{\ell}^{j \star}_{-j}$.
\\
This contradiction indeed shows that $g$ has to be link deletion proof relative to $\varphi^a$.
\end{proof}

\medskip\noindent
The proof of Theorem \ref{9:equiv:mon-sps} now proceeds as follows.
\\
First we show that strict pairwise stability for $\varphi^a$ implies monadic stability in $(A^a, \pi^a )$ under the hypothesis that $c \gg 0$.
\\
Let $g \in \mathbb{G}^N$ be a network that is strictly pairwise stable with regard to the network payoff function $\varphi^a$ as given in the assertion. Then $g$ is strong link deletion proof and satisfies the property that
\[
ij \notin g \mbox{ implies that } \varphi^a_i (g+ij) < \varphi^a_i (g) \mbox{ as well as } \varphi^a_j(g+ij) < \varphi^a_j (g).
\]
Hence, this can be rewritten as
\begin{equation}
\label{9:eq:main-mon-sps}
ij \notin g \mbox{ implies } \varphi_i (g+ij) - c_{ij} < \varphi_i (g) \mbox{ as well as } \varphi_j (g+ij)-c_{ji} < \varphi_j (g).
\end{equation}
With $g$ we define for all $i \in N \colon$
\begin{align*}
\hat{\ell}_{ij} = 1 & \mbox{ if } ij \in g \\
\hat{\ell}_{ij} = 0 & \mbox{ if } ij \notin g
\end{align*}
Hence, $g (\hat{\ell}) =g$ and $\hat{\ell}$ is non-superfluous. We now investigate whether the given communication profile $\hat{\ell}$ is indeed a best response to the monadic belief system $\hat{\ell}^{i \star}$ for all $i \in N$ as required by the definition of weak monadic stability.
\begin{description}
\item[Case A:] $ij \notin g$.
\\
From (\ref{9:eq:main-mon-sps}) it follows immediately that $\hat{\ell}^{i \star}_{ji} = \hat{\ell}^{j \star}_{ij} =0$. From the hypothesis that $c_{ij} >0$ and $c_{ji} >0$ and the definition of monadic belief systems, it follows with Lemma \ref{9:lem:61} that $\hat{\ell}_{ij} =0$ is the unique best response to $\hat{\ell}^{i \star}_{-i}$ and that $\hat{\ell}_{ji} =0$ is the unique best response to $\hat{\ell}^{j \star}_{-j}$.
\\
Hence, for Case A the communication strategy $\hat{\ell}$ satisfies the condition of weak monadic stability.

\item[Case B:] $ij \in g$.
\\
In this case $\hat{\ell}_{ij} = \hat{\ell}_{ji} =1$. Link deletion proofness of $g$ now implies that $\hat{\ell}^{i \star}_{ji} =1$ or else (\ref{9:eq:main-mon-sps}) is contradicted.
\end{description}
Cases A and B now imply that
\begin{equation}
\label{9:eq:second-mon-sps}
ij \in g \mbox{ if and only if } \hat{\ell}^{i \star}_{ji} = \hat{\ell}^{j \star}_{ij} =1 .
\end{equation}
Applying strong link deletion proofness and the insight for Case A leads us to the conclusion that $\hat{\ell}_i$ is indeed the unique best response to $\hat{\ell}^{i \star}_{-i}$. This in turn implies that $\hat{\ell}$ supports $g$ as a weakly monadically stable network.
\\
Finally, it is immediately clear from (\ref{9:eq:second-mon-sps}) and the definition of $\hat{\ell}$ that for all $i,j \in N \colon \hat{\ell}^{i \star}_{ji} = \hat{\ell}_{ij}$, implying that the monadic beliefs are indeed confirmed.
\\
Thus, we conclude that $\hat{\ell}$ supports $g$ as a monadically stable network. This completes the proof of the first part of the assertion.

\bigskip\noindent
Second, we show that the monadic stability of a network for $(A^a, \pi^a )$ implies strict pairwise stability for $\varphi^a$ under the hypothesis that $c \gg 0$.
\\
Let $g \in \mathbb{G}^N$ be monadically stable. Then there exists some action tuple $\hat{\ell} \in A^a$ such that $g = g(\hat{\ell})$ and for every player $i \in N \colon \hat{\ell}_i \in A^a_i$ is a best response to $\hat{\ell}^{i \star}_{-i}$ for the payoff function $\pi^a$. Furthermore, $\hat{\ell}^{i \star}_{-i} = \hat{\ell}_{-i}$.
\\
From Lemma \ref{9:lem:63} we already know that $g$ has to be link deletion proof for $\varphi^a$ since $g$ is weakly monadically stable. Hence, for every $ij \in g$ we have that $\varphi_i (g-ij) +c_{ij} \leqslant \varphi_i(g)$. Now through the definition of the monadic belief systems and the self-confirming condition of monadic stability we conclude that for every $ij \in g$:
\begin{equation}
\hat{\ell}_{ij} = \hat{\ell}^{j \star}_{ij} = \hat{\ell}_{ji} = \hat{\ell}^{i \star}_{ji} =1 .
\end{equation}
Let $i \in N$ and $h \subset L_i (g)$. Now we define $\ell^h \in A^a_i$ by
\[
\ell^h_{ij} = \left\{
\begin{array}{ll}
\hat{\ell}_{ij} & \mbox{if } ij \notin h \\
0 & \mbox{if } ij \in h .
\end{array}
\right.
\]
Then $g \left( \, \ell^h , \hat{\ell}_{-i} \, \right) = g -h$. Since $\hat{\ell}_i$ is a best response to $\hat{\ell}^{i \star}_{-i} = \hat{\ell}_{-i}$ it has to hold that\footnote{Here we again apply the confirmation condition for monadic stability that is satisfied by $\hat{\ell}$.}
\[
\pi^a_i \left( \, \ell^h , \hat{\ell}_{-i} \, \right) \leqslant \pi^a_i ( \hat{\ell}) .
\]
Hence,
\begin{equation}
\varphi_i ( g-h ) + \sum_{ij \in h} c_{ij} \leqslant \varphi_i (g) .
\end{equation}
This in turn implies that $\varphi^a_i (g-h) \leqslant \varphi^a_i (g)$.
\\
Since, $i \in N$ and $h$ are chosen arbitrarily, the network $g$ has to be strong link deletion proof.
\\
Next, let $ij \notin g$. Then $\hat{\ell}_{ij} = 0$ and/or $\hat{\ell}_{ji} = 0$. Suppose that $\hat{\ell}_{ji} = 0$. Then by the confirmation condition of monadic stability it follows that $\hat{\ell}^{i \star}_{ji} = \hat{\ell}_{ji} = 0$. Hence by Lemma \ref{9:lem:61}, $\hat{\ell}_{ij} = 0$. Thus we conclude that for every $ij \notin g$:
\begin{equation}
\hat{\ell}_{ij} = \hat{\ell}^{j \star}_{ij} = \hat{\ell}_{ji} = \hat{\ell}^{i \star}_{ji} = 0.
\end{equation}
This in turn implies through the definition of the monadic belief system that $\varphi_i (g+ij)-c_{ij} < \varphi_i (g)$ as well as $\varphi_j (g+ij)-c_{ji} < \varphi_j (g)$. Or $\varphi^a_i (g+ij) < \varphi^a_i (g)$ as well as $\varphi^a_j (g+ij) < \varphi^a_j (g)$. This shows the assertion that $g$ is indeed strictly pairwise stable.

\medskip\noindent
This completes the proof of Theorem \ref{9:equiv:mon-sps}.

\end{document}